\documentclass[11pt]{article}
\usepackage{fullpage, tikz}
\usepackage[normalem]{ulem}
\usetikzlibrary{positioning}
\usepackage[T1]{fontenc}

\usepackage{geometry}
\newcommand*{\TCNLA}[1][]{}  %
\newcommand*{\ifNoAppendix}[1]{}
\newcommand{\probb}[1]{\mathbb{P}\left[#1\right]}

\usepackage{xcolor}
\usepackage{dirtytalk}
\usepackage{enumitem}
\usepackage{amsfonts}
\usepackage{amsthm}
\usepackage{amsmath}
\usepackage{amssymb}
\usepackage{mathtools}
\usepackage{bbm}
\usepackage[hidelinks]{hyperref}
\usepackage{cleveref}
\usepackage{enumerate}
\usepackage{comment}
\usepackage{algorithm}
\usepackage[noend]{algpseudocode}

\usepackage{natbib}
\usepackage{thmtools}
\usepackage{thm-restate}
\usepackage{tikz}
\usetikzlibrary{arrows.meta,patterns,calc}

\renewcommand{\cite}[1]{\citep{#1}}
\newcommand{\textcite}[1]{\citet{#1}}

\newtheorem{theorem}{Theorem}[section]

\newtheorem{lemma}[theorem]{Lemma}

\newtheorem{corollary}[theorem]{Corollary}
\newtheorem{proposition}[theorem]{Proposition}

\theoremstyle{definition}
\newtheorem{definition}{Definition}

\newenvironment{proofsketch}
  {\begin{proof}[Proof sketch]}
  {\end{proof}}

\newcommand{\I}{{\mathcal I}}
\newcommand{\E}{{\mathbb{E} }}

\newcommand{\MMS}{{\sf{MMS}}}

\let\oldnl\nl
\newcommand{\nonl}{\renewcommand{\nl}{\let\nl\oldnl}}

\renewcommand{\nonl}{\renewcommand{\nl}{\let\nl\oldnl}}
\long\def\symbolfootnote[#1]#2{\begingroup%
\def\thefootnote{\fnsymbol{footnote}}\footnote[#1]{#2}\endgroup}

\newcommand{\poly}{\operatorname{poly}}
\newcommand{\Pot}{\Pi}

\renewcommand{\arraystretch}{1.5}

\newcommand{\rh}[1]
{{\textcolor[rgb]{0,0,0}{#1}}}

\newcommand{\swr}{\textsc{SwR}}

\newcommand{\iswr}{\textsc{IterSwR}}
\newcommand{\ws}[1][i]{\mathcal{S}_{#1}}
\newcommand{\argmax}{\operatorname{argmax}}

\newcommand{\spc}{{\mathcal S}}
\newcommand{\nspc}{{\mathcal N}}
\newcommand{\shadows}{\hat{\mathcal S}}

\newcommand{\A}{\mathcal A}

\usepackage{booktabs}
\usepackage{adjustbox}

\title{Single-or-Sample: Online Fair Allocation for Combinatorial Agents}

\author{
Shuchi Chawla\thanks{Author names appear in alphabetical order.} \\
\small University of Texas, Austin \\
\small \texttt{shuchi@cs.utexas.edu}
\and
Zhiyi Huang \\
\small University of Texas, Austin \\
\small \texttt{zhiyih@cs.utexas.edu}
\and
Pooja Kulkarni \\
\small University of Chicago \\
\small \texttt{kulkarnip@uchicago.edu}
\and
Ruta Mehta \\
\small University of Illinois Urbana-Champaign \\
\small \texttt{rutameht@illinois.edu}
\and
Parnian Shahkar \\
\small University of California, Irvine \\
\small \texttt{shahkarp@uci.edu}
}

\date{}

\begin{document}

\maketitle

\begin{abstract}
We study the fair allocation of \(m\) indivisible goods among \(n\) heterogeneous agents arriving online, under maximin share (MMS) fairness. The key challenge lies in making allocation decisions in an adversarial setting without knowing future agents' valuations. No guarantees were known prior to our work in adversarial settings.

Our main positive result is an online allocation algorithm for submodular agents that with probability at least $0.4$ guarantees a $1/24$-approximation to MMS to all agents simultaneously. Our algorithm intersperses a greedy-style adaptive sampling-without-replacement protocol with single-good reduction rounds. For XOS valuations, we adapt this single-or-sample approach to guarantee a $c$-MMS for all agents simultaneously with probability $1-O(c)$ for any parameter $c\in O(1/\log n)$. A key technical contribution of our work is in establishing Chernoff-style tail bounds on item allocation events for this adaptive sampling process despite strong positive correlations among the items, a result that may be of independent interest.

The approximation versus success probability tradeoffs achieved by our algorithms are asymptotically tight. We show that even for binary additive valuations, no randomized algorithm can ensure $c$-MMS to all agents with probability exceeding $(1-c^2/5)$, for any $c\in (0,1]$. For XOS valuations, stronger lower bounds hold: no randomized algorithm can ensure $c$-MMS to all agents with probability exceeding $(1-c/3)$, for any $c\in (0,1]$. Furthermore, no algorithm can achieve $\omega(1/\log \log n)$-MMS with any constant probability. 
All our lower bounds hold even for simple binary instances whose MMS values and MMS partitions can be computed in linear time.
\end{abstract}
\thispagestyle{empty}

\newpage
\thispagestyle{empty}
\vspace{1cm}
\setcounter{tocdepth}{1} 
\tableofcontents
\newpage

\setcounter{page}{1}
\section{Introduction}
\label{sec:intro}

Consider a disaster relief supply center that distributes essential goods to members of an affected community.  Recipients arrive over time, communicate their needs, and receive allocations. The center must manage its limited supply carefully: allocating too generously to early arrivals risks shortages for those who come later, while overly conservative allocations can lead to inefficiency and unmet needs. Compounding this challenge, 
future demand is difficult to forecast, as the recipients may have highly heterogeneous needs depending on their circumstances and the local impact of the disaster. Nevertheless, to ensure social stability, it is critical to ensure that in such scenarios the resulting allocation is {\em fair to all recipients in hindsight}. We therefore investigate the natural question: 
in an uncertain, dynamic environment, is it possible to design allocation policies that achieve meaningful guarantees of fairness and efficiency?

Motivated by this, we study an online allocation problem with $m$ indivisible goods and $n$ heterogeneous agents with combinatorial valuation functions. Agents arrive sequentially over time; upon arrival, an agent receives an allocation and then departs, and all allocation decisions are irrevocable. The allocator knows the total number of agents 
$n$ in advance, but has no information about the valuation functions of future arrivals. The goal is to design an online allocation policy that guarantees a fair share to {\em every} agent, despite the uncertainty inherent in future demand.

In this work, we focus on a notion of fairness that depends only on an agent's own valuation and allocation, rather than on allocations made to other agents, i.e., a share-based notion of fairness. For indivisible goods, the most natural and sought-after such notion is the {\em maximin share (MMS)}  \cite{budish2011combinatorial}. Informally, the MMS of an agent is the maximum value she can guarantee for herself by partitioning all goods into $n$ bundles and receiving the least desirable bundle. Formally, let $v_i:2^{[m]}\rightarrow \mathbb{R}_+$ denote the monotone valuation function of agent $i$ for $i\in [n]$. The MMS value of agent $i$, denoted $\MMS_i$, is defined as \[\MMS_i=\max_{(P_1,\dots,P_n)\in \Pi_n([m])} \min_{\ell\in[n]} v_i(P_\ell),\] 
where $\Pi_n([m])$ denotes the set of all partitions of the set of goods $[m]$ into $n$ bundles. 

The problem of finding an MMS allocation has been studied extensively in the offline setting: in light of early non-existence results \cite{procaccia2014fair}, the focus has been on finding an approximate MMS allocation; in an $\alpha$-MMS allocation each agent $i$ receives a bundle worth at least $\alpha \cdot \MMS_i$. For the cases of additive and submodular valuation functions, constant-factor MMS guarantees are well-established \cite{Aziz2017, barman2020approximationalgorithmsmaximinfair,akrami2024breaking, huang2025fptas,heidari2026improved, uziahu2023fairallocationindivisiblegoods}. 

In contrast to these offline results, despite the recent focus on online fair division, online MMS remains relatively unexplored, in part due to inherent technical challenges \cite{amanatidis2022fair,kulkarni2025online}. No MMS guarantees are known for online good arrivals; the main barrier is that MMS value is undefined until all the goods are seen. For online agent arrivals, prior work \cite{kulkarni2025online} introduced a stochastic model of online fair division where the agents are drawn from a distribution over a fixed set of $k$ a priori known valuation functions. However, in many settings it is not possible to predict the set of valuations or it is prohibitively expensive to do so. Our work addresses this fundamental gap and asks:
\begin{center}
\begin{minipage}{0.75\textwidth}
\noindent\itshape
Is it possible to achieve any nontrivial approximate-MMS guarantee when the allocator has no prior knowledge of the agents' valuation functions?
\end{minipage}
\end{center}


In this paper, we answer this question affirmatively for submodular and XOS valuations against an oblivious adversary\footnote{An oblivious adversary can select worst-case agents' valuation functions after knowing the algorithm, but {\em not} knowing the outcomes of the random coin tosses (decisions) made by the algorithm. Note that against an adaptive adversary in the no future information setting, \cite{kulkarni2025online}'s strong impossibility of online fair allocation continues to hold -- by which we cannot ensure any non-trivial approximation to MMS.}. We provide 
upper and lower bounds on the approximation factors achievable for agents with additive, submodular, and XOS (fractionally subadditive) valuations (see Table \ref{tab:res} for a summary of our results). 
First, we observe the following two tradeoffs between the approximation guarantee and success probability of any algorithm: (i) For even binary additive valuations, no algorithm can guarantee a $c$-approximate MMS with a probability better than $(1-c^2/5)$ for any $c \in [0,1]$ (Theorem \ref{thm:lb_c}). (ii) For XOS valuations, no algorithm can guarantee a $c$-approximate MMS with probability better than $(1-c/3)$ for any $c \in [0,1]$ (Theorem \ref{thm:lb_xos_c}).
In the light of these tradeoffs, {\bf we focus on achieving constant factor MMS approximation, simultaneously for all agents, with constant probability.} 

Our main result (\Cref{thm:submod-rr-main}) shows that such a result can be achieved for submodular valuations. We design a randomized greedy-like algorithm that uses an iterative, adaptive sampling-without-replacement process interspersed with single-good reduction rounds. A key technical contribution is our analysis of the adaptive sampling process 
which may be of independent interest in other online allocation problems---we show that the process lends itself to Chernoff-style tail bounds {\em despite exhibiting positive correlations} and thereby ruling out the usual routes to concentration (\Cref{lem:iterative-cylinder}). In contrast to this, for XOS valuations, we prove that no randomized algorithm can achieve better than a $2/\log \log n$  approximation, if we insist on succeeding with constant probability (\Cref{thm:xos-lower-bound}). In fact, this lower bound holds even under the {\em no-large-goods} assumption. Such an assumption, for submodular valuations facilitates a much stronger guarantee, namely constant MMS with $(1-1/\poly(n))$ probability (\Cref{thm:sub-smallVal}). On the positive side, we give a tight tradeoff for XOS valuations achieving a $c$-MMS with probability $(1-O(c))$ for any $c \in [0, 1/(4\log n)]$.
\begin{table}[ht!]
    \centering
    \small
    \setlength{\tabcolsep}{5pt}
    \renewcommand{\arraystretch}{1.35}

\setlength{\tabcolsep}{4pt}
\renewcommand{\arraystretch}{1.4}

\begin{adjustbox}{max width=\linewidth}
\begin{tabular}{|c|c|c|c|@{}|}
\hline
\textbf{Valuation function}
    & \textbf{Lower bound}
    & \textbf{Upper bound}
    & \textbf{Reference} \\
\hline \hline

Binary additive
    & $\bigl(c,1-\frac{c^2}{5}\bigr),\; c\in[0,1]$
    & 
    & Theorem~\ref{thm:lb_c} \\
\hline
Submodular; no large items
    & ---
    & $\bigl(\frac{1}{12},1-ne^{-n/10}\bigr)$
    & Theorem~\ref{thm:sub-smallVal} \\
\hline
Submodular
    & $\bigl(c,1-\frac{c^2}{5}\bigr),\; c\in[0,1]$
    & $\bigl(\frac{1}{24},0.423\bigr)\;\uparrow$
    & Theorem~\ref{thm:submod-rr-main} \\

\hline

Binary XOS$^*$; no large items
    & $\bigl(\frac{2}{\log\log n},c\bigr),\; c\in\Theta(1)$
    & 
    & Theorem~\ref{thm:xos-lower-bound} \\
\hline
Binary XOS$^*$
    & $\bigl(c,1-\frac{c}{3}\bigr),\; c\in[0,1]$
    & 
    & Theorem~\ref{thm:lb_xos_c} \\
\hline
XOS
    & $\bigl(\frac{2}{\log\log n},c\bigr),\; c\in\Theta(1)$;
      $\bigl(c,1-\frac{c}{3}\bigr),\; c\in[0,1]$
    & $\bigl(c,1-O(c)\bigr)\;\uparrow,\;
       c\le\frac{1}{4\log n}$
    & Theorem~\ref{thm:xos-alg} \\

\hline
\end{tabular}
\end{adjustbox}

    \caption[A summary of our results.]{A summary of our results. A lower bound of $(\alpha,\beta)$ indicates that no randomized algorithm can achieve $\alpha$-MMS with probability exceeding $\beta$. An upper bound of $(\alpha, \beta)$ indicates a randomized algorithm achieving $\alpha$-MMS simultaneously for all agents with probability at least $\beta$. Constants have not been optimized. An up arrow ($\uparrow$) indicates that the result also applies to preceding rows. All our algorithms require normalized input. Thus, when an agent arrives, we need to know their MMS value to normalize the instance (Appendix \ref{app:norm}); for an XOS agent, we also need their MMS partition. The algorithms can be implemented efficiently using a value oracle at a small constant factor loss (Section~\ref{sec:comp}). $(^*)$ All lower bound XOS instances created are {\em simple}, that is, each XOS valuation function is binary, has at most $O(n)$ clauses, and their MMS value and MMS partition can be computed in $O(m+n)$ time.
    }
    \label{tab:res}
\end{table}

Besides the fairness guarantee, our algorithm for submodular valuations has several nice features. First, it requires {\em limited communication} between the agents and the allocator. In particular, the allocator asks the agent to report her most desirable items over a linear number of rounds of allocation, and never elicits her full valuation function. Second, our algorithm exhibits strong {\em incentive guarantees}. In particular, 
the MMS approximation guaranteed to an agent holds irrespective of the strategic behavior of other agents. If an agent behaves strategically or adversarially, she may harm her own utility but not that of other truthful agents. 
Third, it exhibits additional sought-after guarantees with much higher probabilities (\Cref{sec:additional}): $(i)$ each agent gets constant MMS with probability polynomially close to one, $(ii)$ all but $O(\log n)$ randomly chosen agents get constant MMS, again with probability polynomially close to one, and $(iii)$ we can provide a constant $1$-out-of-$(n+o(n))$ MMS guarantee to all agents simultaneously with probability $(1-1/n)$. 


We now give an overview of our algorithmic techniques and analysis that lead to these results.  


\subsection{Technical Overview}
To form intuition for our algorithm, consider an additive agent. At a high level, when an agent arrives, our goal is to allocate about a $1/n$ fraction of her total value while ensuring that no individual item is allocated with high probability. In the absence of information about which items future agents will value more than others, we strive to make the allocation distribution as symmetric across items as possible.

A natural first attempt is to allocate each item to the agent independently with $1/\alpha n$ probability for some constant $\alpha\ge1$. Then every item has a constant probability of surviving until the last agent arrives, and each agent obtains in expectation a constant fraction of $1/n$ times her total value, an upper bound on her MMS value for additive valuations. This does not, however, give a simultaneous guarantee. For example, suppose all agents value each of $\alpha n$ items at $1$. Each agent receives $\Theta(1)$ items in expectation, yet every round has a constant conditional probability of allocating nothing. Consequently, some agent receives no items with probability $1-\exp(-\Omega(n))$.
The problem is that the total value received by an agent has large variance. 

To control the variance, we would like to fix the number of items sampled while preserving the uniform allocation probabilities. Fixing only the overall sample size does not suffice: when item values vary widely, samples of the same size can have very different values. This suggests controlling the number of items selected separately at each value scale.

\paragraph{Bucketing and sampling without replacement.}
When an agent arrives, consider partitioning the remaining items into buckets according to geometrically increasing value ranges. Ignoring rounding, from each bucket we select a uniformly random $1/\alpha n$ fraction of the bucket. Because items within a bucket have comparable values, fixing the number selected ensures that, conditioned on the remaining items, the agent deterministically receives an $\Omega(1/n)$ fraction of their total value, while every item is selected in each round with probability at most $1/\alpha n$; summing over at most \(n\) rounds, its probability of ever being selected is at most \(1/\alpha\).

\paragraph{Submodular valuations.}
For {\em submodular valuations}, the bucketing approach above does not apply directly, because the marginal value of an item depends on which other items the agent has already received. Instead, we develop a random greedy allocation procedure that continues to allocate each item with probability only $1/\alpha n$. When an agent arrives, we identify the $\alpha n$ unallocated items with the largest marginal values, allocate one of them uniformly at random, and reserve the other $(\alpha n-1)$ for future agents. We then recompute marginal values over the remaining items with respect to the agent's current allocation and repeat until all items have been considered. We show that this protocol guarantees the agent an \rh{$\Omega(1/n)$ fraction of the sum of the residual values of her MMS-defining bundles}.


\paragraph{Bounding exponential moments of a suitable potential.}
To finish off the analysis, we must  show that enough value survives for later agents with high probability. 
A basic guarantee we seek is that, for any fixed set $S$ of items, a constant fraction survives with probability $1-\exp(-\Omega(|S|))$; we also need weighted versions of this guarantee to track the agents' residual values. The bound of $1/(\alpha n)$ on each item's allocation probability in every round ensures that each item survives with constant probability, and hence that $\Omega(|S|)$ items survive in expectation. The challenge is to strengthen this expectation bound to a concentration guarantee. Small itemwise probabilities alone cannot do so: allocating all remaining items together with probability $1/(\alpha n)$ in each round satisfies the same bounds, yet leaves no survivors with constant probability over $n$ rounds. Nor can we simply invoke negative association, which need not survive adaptive bucket choices or changing partitions across agents; see Appendix~\ref{app:negassoc}.

We exploit the additional structure of sampling without replacement: within each iteration, the buckets are disjoint, and each contributes a uniform sample of prescribed size. For an iteration with sampling rate at most $p$, we show that the probability of selecting every item in any fixed set $T$ is at most $p^{|T|}$, even when later buckets depend on earlier samples. This joint inclusion bound composes across iterations and yields moment generating function (MGF) bounds for fixed nonnegative weighted sums, which in turn imply the desired concentration. We give elementary inductive proofs of these bounds; the concentration step can also be formulated using multiplicative supermartingale arguments, as in~\cite{KuszmaulQi2021}.

To connect these bounds to submodular value, we fix an agent's MMS partition and assign additive supporting weights within each bundle. After normalizing the agent's MMS to one, we normalize each bundle's supporting weights to total one. The surviving weight is then a linear potential, initially $n$, that lower-bounds the sum of the residual bundle values. Concentration keeps this potential high even if some bundles are severely depleted and the residual MMS falls to nearly zero. The random greedy guarantee translates this surviving potential into a valuable allocation.


\paragraph{High-valued items.}
The approach above is insufficient when an agent's MMS value is concentrated on a few high-valued items. We handle this case with the standard {\em single-item reduction} step: any agent who can obtain a good approximation to her MMS value from a single item is allowed to take such an item immediately. Any agent who uses such a step is called a ``special'' iteration. To account for the potential large loss in a special round, the algorithm keeps track of the number $t$ of these ``special'' iterations and adjusts the allocation probabilities in subsequent ``non-special'' rounds accordingly, replacing $n$ by $n-t$. A more concerning complication is that interspersing sampling rounds with special rounds, whose single-item allocations may be chosen adversarially in response to earlier random outcomes, introduces further dependencies among item allocation events rendering our MGF bound potentially unusable. Moreover, a long sequence of such reductions can leave the final few agents with a constant probability of failure. Nevertheless, we show that the overall protocol simultaneously satisfies all agents with constant probability.

Towards this, for submodular agents, we measure value loss relative to the same benchmark $n-t$ that governs the sampling rate. Our analysis maintains $n-t$ active MMS bundles, discarding one after each special round and choosing the bundle containing the allocated item whenever possible. We define the \emph{deficit} as the shortfall of their surviving supporting weight from this adjusted benchmark, with each original bundle's supporting weight normalized to one. A special round reduces the benchmark by one while removing at most one unit of supporting weight, so it cannot increase the deficit. Thus, although the active bundles change adaptively, only sampling rounds can increase the deficit, and our conditional MGF bound controls these increases.


\paragraph{XOS valuations and a shadow process.}

For an XOS agent, spreading an allocation across MMS bundles can be ineffective: its value may be only the largest contribution from any one bundle. Our algorithm therefore chooses a uniformly random MMS bundle with sufficiently high residual value and samples within it at rate $\Theta(1/\alpha)$, where $\alpha$ is at least $4\log n$. 
Concentrating the allocation within one bundle benefits the current agent, while randomizing the choice of bundle limits the damage to future agents.

The submodular concentration argument does not carry over directly in this setting, although it is still possible to define a deficit function that is nonincreasing during special rounds. Choosing a single MMS bundle correlates the allocation events of its items: one sampling round can remove substantial supporting weight from many bundles of a future agent, despite small itemwise allocation probabilities. Consequently, the conditional MGF bound used in the submodular analysis no longer applies. Our XOS concentration bound instead controls losses against fixed supporting weights. 

To analyze the special rounds, we introduce a deterministic \emph{shadow process} whose allocations identify a fixed collection of MMS bundles to set aside. Roughly speaking, the process overestimates the number of MMS bundles of the agent affected by special rounds. We can then apply concentration to the remaining reference bundles and use a coupling to control the additional damage from actual special rounds, ensuring that sufficiently many bundles retain enough value.

\paragraph{Organization.} We formalize the online fair allocation problem in \Cref{sec:prelims}, and present some simple tradeoffs between the approximation achievable and the probability of success. In Section~\ref{sec:tools} we state the MGF bounds for sampling without replacement from adaptively defined buckets that we use throughout our analysis. As a warm-up to our final algorithm and analysis, in \Cref{sec:additive} we consider instances where no agent has high-valued items, and focus on presenting the main ideas behind our sampling algorithm and probabilistic analysis. In \Cref{sec:submod} we present the complete algorithm and analysis for submodular valuations to achieve constant-MMS with constant probability. \Cref{sec:additional} shows the additional guarantees our algorithm exhibits in terms of individual guarantees, guarantees for almost all agents, and $1$-out-of-$d$ guarantees. Finally, \Cref{sec:xos-lower-bound} focuses on XOS valuations, and presents the algorithm achieving $(1/\log n)$-MMS. A stronger lower bound for XOS valuations ruling out $(2/\log \log n)$-MMS with probability better than $o(1)$ is presented in Appendix~\ref{app:xos-lb}. Additional related work is discussed in \Cref{sec:related_work}.

\section{
Preliminaries, Basic Tradeoffs, and Concentration Tools}\label{sec:prelims}

{\bf Notation.} For a positive integer $n$, let $[n] = \{1, 2, \ldots, n\}$, and for two positive integers $i<j$, let $[i,j] = \{i, i+1, \cdots,j\}$. For a set or list \( N \), let $|N|$ denote the number of elements in $N$.

\subsection{Problem Setting}
\noindent \textbf{Instance.}
We consider the problem of fairly allocating $m$ indivisible goods among $n$ heterogeneous agents. Throughout, we use index $i$ to denote an agent and index $j$ to denote a good. Without loss of generality, we assume agent $i$ arrives at time $i$, and her preferences over the subsets of goods are represented by a non-negative non-decreasing valuation function $v_i : 2^{[m]} \rightarrow \mathbb{R}_+$. We assume that valuations are normalized such that $v_i(\emptyset) = 0$ for all $i\in [n]$. We denote an instance of the fair division problem as $\I = \{[n], [m], (v_i)_{i \in [n]}\}$. 
\medskip

\noindent{\bf Valuation functions.} In this paper, we study three well-known valuation functions, namely additive, submodular, and fractionally subadditive (a.k.a. XOS) values. 

A valuation function $v_i$ is said to be {\em additive} if for any subset $S\subseteq [m]$ of goods, $v_i(S)=\sum_{j\in S} v_{i}(j)$ where $v_i(j)$ is the value that agent $i$ derives from good $j$. 
%
%
%
{\em Submodular} functions are much more general, and capture the {\em diminishing marginal utility} of agents' preferences \cite{kauder2015history}.  
Formally, a valuation function $v_i$ is said to be submodular iff it satisfies the following:
\[
\forall S\subset T\subseteq [m] \mbox{ and } \forall j\in [m]\setminus T,\ \ \ v_i(T\cup \{j\}) - v_i(T) \le v_i(S\cup\{j\}) -v_i(S)
\]
{\em Fractionally subadditive (XOS)} functions further generalize submodular functions  \cite{xos}. An XOS function $v_i$ is the maximum over $l$ additive clauses $v_{ik}$, $k\in [l]$, where $l\in\mathbb{Z}^+$: 
\[ v_i(S)=\max_{k=1}^l v_{ik}(S). 
\]
Note that, reporting a submodular or an XOS valuation function requires the agent to report $2^m$ values in general. Our algorithms instead require only limited information about the valuation function and have polynomial communication complexity. See Section~\ref{sec:comp} for further discussion of computational efficiency.

\medskip

\noindent{\bf Fairness Notion: Maximin Share (MMS).}
For any set \(S\) of goods and a positive integer \(d\), let \(\Pi_d(S)\) denote the collection of all partitions of \(S\) into \(d\) bundles. 
\begin{definition}[Maximin Share (MMS)]\label{def:mms}
Given an instance $\I = ([n], [m], \{v_i\}_{i\in[n]})$ of the fair division problem, the \emph{maximin share} (MMS) value of agent $i$ is defined as
\[
\MMS_i = \max_{P \in \Pi_n([m])} \min_{k \in [n]} v_i(P_k).
\]
\end{definition}
An allocation $(A_1,\dots,A_n)\in \Pi_n([m])$ is said to be an MMS allocation if $v_i(A_i) \ge \MMS_i$ for all agents $i$. It is said to be $\alpha$-approximate MMS or just $\alpha$-MMS, for an $\alpha\in[0,1]$ if, \[ v_i(A_i) \ge \alpha \cdot \MMS_i,\ \ \ \forall i\in[n].\]

\noindent{\bf \textsc{Online MMS} with adversarial agent arrival.}
We study the problem of computing an (approximate) MMS allocation in the online setting where agents arrive online over time. When agent $i$ arrives at time $i$, her valuation function is revealed to the algorithm, and it must irrevocably allocate a subset of goods to the agent from the available set of goods. 

We consider a standard oblivious adversary model where the adversary has access to the algorithm we use to allocate the goods to agents, but not to the random coin tosses of the algorithm. The adversary can use this information for choosing the agents' valuations. It is easy to see that against an adaptive adversary, who has access to the algorithm's randomness as well, no non-trivial guarantees are possible \cite{kulkarni2025online}. 

\begin{definition}[\((\alpha, \beta)\)-competitive algorithm]
For $\alpha \in [0,1]$ and $\beta \in [0,1]$, we say that an online algorithm is $(\alpha, \beta)$-$\MMS$ competitive if it ensures an $\alpha$-approximate MMS allocation to {\em all} agents ex-post, with probability at least $\beta$. 
\end{definition}


\noindent{\bf Normalization.} 
To aid our analysis we will assume that the given instance is normalized. Normalization holds without loss of generality (see Appendix \ref{app:norm}).


\begin{definition}[Normalized Instance]\label{def:norm}
An instance $\I$ of \textsc{OnlineMMS} is \emph{normalized} if for every agent $i \in [n]$, $\MMS_i = 1$. For an MMS partition $(A_1,\dots,A_n)$ of agent $i$ it follows that $v_i(A_k)\ge 1$ for each $k\in[n]$.
\end{definition}



\subsection{An Approximation vs Success Probability Tradeoff}
\label{sec:lower}
To set the stage for our upper bounds, we first show that no algorithm, randomized or not, can achieve a $(c, (1-\frac{c^2}{5}))$ approximation for any $c>0$, i.e., achieving $c$-MMS necessarily entails an $\Omega(c^2)$ failure probability. In fact, for XOS agents, the failure probability is even larger---$\Omega(c)$. 

\begin{theorem}
\label{thm:lb_c}
    For any $c\in(0,1]$, no randomized algorithm can achieve \((c, \beta)\)-competitiveness with $\beta\geq1-c^2/5$ against an oblivious adversary, even with binary additive valuations. 
\end{theorem}

\begin{proof}
     By Yao's minimax theorem, it suffices to construct a distribution over instances where every deterministic algorithm obtaining $c$-\MMS\ succeeds with probability less than $1-\frac{c^2}{5}$. The construction is as follows. Let $n$ be a sufficiently large integer. The instance contains $n$ agents and $m=\lceil\frac{1}{c-\epsilon}\rceil n$ goods where $\epsilon>0$ is a small constant. Let $S$ be a uniformly random subset of $[m]$ of size $n$. The first agent values every good at $1$; her MMS value is $\lceil\frac{1}{c-\epsilon}\rceil$. The remaining $n-1$ agents value each good in $S$ at $1$ and goods in $[m]\setminus S$ at $0$; their MMS values are $1$. 

    Fix a deterministic online algorithm. If it fails to give the first agent a $c$-MMS bundle, its success probability is $0$. Therefore, it must allocate at least $c\cdot\lceil\frac{1}{c-\epsilon}\rceil>1$ goods to the first agent. Since $|S|=n$, to ensure the remaining agents each get at least one good from $S$, the algorithm can  allocate at most one good from $S$ to the first agent. Thus the success probability $\beta$ satisfies
    \begin{align*}
    \beta &\le1-\probb{\text{$S$ contains the first $2$ goods allocated to the first agent}}\\
    &=1-\binom{m-2}{n-2}/\binom{m}{n} = 1-\frac{n(n-1)}{m(m-1)}
    <1-\frac{c^2}{5}
    \end{align*}
    where we chose $n$ large enough and $\epsilon$ small enough and used $m/n=\lceil\frac{1}{c-\epsilon}\rceil$.
\end{proof}


\begin{theorem}
\label{thm:lb_xos_c}
    For any $c\in(0,1]$, no randomized algorithm can achieve \((c, \beta)\)-competitiveness with $\beta\geq1-c/3$ against an oblivious adversary, even with binary XOS valuations.
\end{theorem}

\begin{proof}
     Again by Yao's minimax theorem, it suffices to construct a distribution over instances where every deterministic algorithm succeeds with probability less than $1-\frac{c}{3}$. The construction is similar to the previous one. The instance contains $n$ agents and $\lceil\frac{1}{c-\epsilon}\rceil n$ goods where $\epsilon>0$ is a small constant and $n$ is a sufficiently large multiple of $\lceil\frac{1}{c-\epsilon}\rceil$. The first agent partitions the items into $n$ groups, where each group contains $\lceil\frac{1}{c-\epsilon}\rceil$ items. The agent values every item at $1$. The value of a set is the maximum over the groups of the number of items the set contains from that group. That is, values extended additively over any individual group and the agent is unit demand over different groups. The agent's MMS value is $\lceil\frac{1}{c-\epsilon}\rceil$. The remaining $n-1$ agents have identical additive valuation functions. We choose uniformly randomly $n/\lceil\frac{1}{c-\epsilon}\rceil$ groups out of the $n$ groups. The remaining $n-1$ agents then value every item inside the chosen groups at $1$---note that there are exactly $n$ such items---and all other items at $0$. 

    Fix a deterministic online algorithm. If its first allocation fails to give the first agent a $c$-MMS bundle, its success probability is zero. Therefore, it must allocate at least $c\cdot\lceil\frac{1}{c-\epsilon}\rceil>1$ goods to the first agent from one single group. Call this group $S$. To ensure the remaining agents each get a $c$-MMS bundle, $S$ cannot be one of the groups that the remaining agents have positive value on, otherwise only $n-2$ items of positive value will remain. Thus we have $\beta\le1-\probb{\text{$S$ is one of the $\frac{n}{\lceil\frac{1}{c-\epsilon}\rceil}$ groups}}=1-\frac{1}{\lceil\frac{1}{c-\epsilon}\rceil}<1-\frac{c}{3}$, where we chose $\epsilon$ to be small enough.
\end{proof}

\subsection{Concentration Tools}
\label{sec:tools}

We will now describe the concentration tools we require for analyzing our algorithms. All of our algorithms entail successive rounds of sampling without replacement from adaptively defined buckets with adaptively chosen sampling probabilities. We begin by describing such a process in its full generality. 

\begin{definition}[Iterative Sampling without Replacement ($\iswr$)]
    Let $G$ be a ground set of $m$ elements. We first formalize simple (single-round) sampling without replacement. Given a probability $p\in [0,1]$, let $P$ be a random variable taking values in $\{\lfloor pm\rfloor, \lceil pm\rceil\}$, with $\mathbb{E}[P]=pm$. We denote by $\swr(G,p)$ the  process that returns a uniformly random subset of $G$ of size $P$.

    The {\em iterative SwR} process over $n$ iterations with parameters $p_1, \cdots, p_n$ proceeds as follows and returns a random sample $A\subset G$.
    \begin{itemize}
    \item Let $G_i$ denote the set of items available at the beginning of iteration $i$. Initialize $G_1=G$ and $A=\emptyset$.
    \item A single iteration $i$ of the process, that we call {\em multi-bucket SwR}, is parameterized by $p_i\in [0,1]$ and returns a random sample $A_i$ of $G_i$. We initialize $A_i=\emptyset$, $\ell=0$, and $G_i^{(0)}=G_i$. While $G_i^{(\ell)}$ is not empty, let $B_i^{(\ell)}$ be any non-empty subset of $G_i^{(\ell)}$, and $A_i^{(\ell)}$ be drawn from $\swr(B_i^{(\ell)},p_i^\ell)$ where $p_i^\ell\le p_i$. Set $G_i^{(\ell+1)}=G_i^{(\ell)}\setminus B_i^{(\ell)}$, $A_i=A_i\cup A_i^{(\ell)}$, and $\ell=\ell+1$, and repeat. Return $A_i$.
    \item We set $G_{i+1}=G_i\setminus A_i$, and $A = A\cup A_i$.
    \item We emphasize that the buckets $B_i^{(\ell)}$ and probabilities $p_i^\ell$ can be chosen adaptively based on the entire history of the sampling process until the $\ell$th round of iteration $i$.
    \item After the $n$th iteration, we return the sample $A$.
    \end{itemize}
\end{definition}

We analyze the iterative sampling process by establishing bounds on the exponential moments of any fixed additive function over the elements in $G$. Consider the process $\iswr(G, p_1, \cdots, p_n)$, and let $q=\prod_{i=1}^n (1-p_i)$. Note that $p_i$ bounds from above the probability of choosing any single element of $G$ in the sample in iteration $i$. $q$ therefore captures the probability that a particular element remains unselected in all iterations. Due to the adaptive nature of the process, elements have complex positive and negative correlations (as we show in Appendix~\ref{app:negassoc}). Nevertheless, we can establish MGF bounds that mirror standard bounds for uncorrelated sampling.

\begin{lemma}
\label{lem:iterative-cylinder}
Consider an iterative SwR process with parameters $p_1,\ldots,p_n$, ground set $G$, and sampled set $A$. Let $q=\prod_{i=1}^n(1-p_i)$.
Let $(x_e)_{e\in G}$ be arbitrary weights with
$x_e\in[0,1]$, $X := \sum_{e\in A} x_e$, and $W := \sum_{e\in G} x_e$. Then, for every $\lambda\ge 0$,
\[
\E[e^{\lambda X}]
\le
\exp\left((1-q)W(e^\lambda-1)\right).
\]
\end{lemma}

We will also need to analyze exponential moments of weighted additive functions for iterative sampling processes, where the weights themselves depend on the preceding history. The MGF bound established in the above lemma allows us to establish such an analysis.
\begin{restatable}{lemma}{condmgf}
\label{lem:conditional-mgf-composition}
Let $(Y_r)_{r=0}^N$ be a real-valued process adapted to a filtration
$(\mathcal{F}_r)_{r=0}^N$, with $Y_0=0$.
Fix $\lambda\ge0$.
For each $r\in\{1,\ldots,N\}$, let $b_r$ be a finite nonnegative
$\mathcal{F}_{r-1}$-measurable random variable.
Suppose that, almost surely,
\begin{equation}
\label{eq:one-round-conditional-mgf}
    \mathbb{E}\left[
        e^{\lambda(Y_r-Y_{r-1})}
        \,\middle|\,\mathcal{F}_{r-1}
    \right]
    \le e^{b_r},
    \qquad r=1,\ldots,N.
\end{equation}
Then, for every $j\in\{0,\ldots,N\}$,
\begin{equation}
\label{eq:composed-conditional-mgf}
    \mathbb{E}\left[
        \exp\left(
            \lambda Y_j-\sum_{r=1}^j b_r
        \right)
    \right]
    \le1.
\end{equation}
The same inequality holds with $j$ replaced by any stopping time
$\tau\in\{0,\ldots,N\}$.
In particular, if for a deterministic constant $B$ we have $\sum_{r=1}^{\tau}b_r\le B$
almost surely, then
$\mathbb{E}[e^{\lambda Y_\tau}]\le e^B$.
\end{restatable}
We defer the proofs of both the lemmas to Appendix~\ref{sec:sampling-ext}.

\section{Warm-Up: Constant MMS in the Absence of Large Values}
\label{sec:additive}

We will now present some of the main ideas underlying our online allocation algorithm. As a warm-up, we will consider the case of additive values under the assumption that no agent places too high a value on a single item.
Recall that each agent's value is normalized so that $\MMS_i=1$ and $v_i([m])\ge n$ for each $i$. We will further assume that $v_i(e)<1/12$ for all $i\in [n], e\in [m]$. 

Our allocation algorithm proceeds as follows. When an agent arrives, she partitions the available items into disjoint sets of size exactly $3n$, possibly leaving some items unassigned. From each such set, the agent receives one item chosen uniformly at random. The algorithm then proceeds to the next agent with the remaining items, and repeats this process. The algorithm is described formally in pseudocode below.

\begin{algorithm}[H]
\caption{Simple Static-Bucketing Allocation}
\label{alg:add-simple}
\begin{algorithmic}[1]
\Require Number of agents $n$, ground set $[m]$, with $v_i(e)<1/12$ for all $i\in [n], e\in [m]$.
\State $R_1 \gets [m]$
\For{$i = 1,2,\dots,n$}
    \State Agent $i$ observes the current remaining set $R_i$
    \State Agent $i$ specifies a partition
    $R_i = B_{i1} \,\dot\cup\, B_{i2} \,\dot\cup\, \cdots \,\dot\cup\, B_{i\ell_i}$, with $|B_{ij}|=3n$ $\,\forall j\in [\ell_i-1]$. \label{step:simple-bucketing}
    \For{$j = 1,2,\dots,\ell_i-1$}
        \State $A_{ij} \gets \swr\!\left(B_{ij}, \frac{1}{3n}\right)$
    \EndFor
    \State Allocate $A_i := \bigcup_{j=1}^{\ell_i-1} A_{ij}$ to agent $i$
    \State $R_{i+1} \gets R_i \setminus A_i$
\EndFor
\end{algorithmic}
\end{algorithm}

Our analysis has two parts. First, the MGF bound established for the $\iswr$ process in Section~\ref{sec:tools} shows that enough value remains for each agent, regardless of the preceding agents' strategies. Second, sorting the remaining goods by value and grouping them into buckets guarantees the agent value at least
$v_i(R_i)/(3n)-1/12$. Together, these bounds yield a $1/12$-MMS allocation for all agents simultaneously, except with probability exponentially
small in $n$.

We first describe the agent's {\em safe strategy}. The agent orders the remaining goods by value and groups consecutive goods into buckets of size $3n$, with a possibly smaller final bucket.

\begin{definition}
    {\bf (Safe Strategy $\ws$ for agent $i$)} Order elements $e\in R_i$ in weakly decreasing order of value $v_i(e)$, and rename them $e_0, e_1, \cdots$ in this order. Define $\ell_i = \lceil |R_i|/3n \rceil$. For $j\in [\ell_i-1]$, set $B_{ij}:=\{e_{3n(j-1)}, \cdots, e_{3nj-1}\}$, and $B_{i\ell_i} := \{e_{3n(\ell_i-1)}, \cdots \}$.
\end{definition}

We now state the two main lemmas underlying our analysis. We provide proof sketches here and defer the full formal argument to Section~\ref{sec:submod} where we prove our fully general
result.

\begin{lemma}
\label{lem:warmup-surviving-weight}
Fix an agent $i\in[n]$ and weights $w_e\in[0,1/12]$ for all goods. Let
$W:=\sum_{e\in[m]}w_e$. Under Algorithm~\ref{alg:add-simple}, regardless of the preceding
agents' partitioning strategies,
\[
    \Pr\!\left[\sum_{e\in R_i}w_e<\frac{7W}{12}\right]
    \le e^{-W/10}.
\]
In particular, for an additive agent $i$, $\Pr[v_i(R_i)<7n/12]\le e^{-n/10}$.
\end{lemma}

\begin{proof}
Let $D_i:=[m]\setminus R_i$ be the goods allocated before agent $i$
arrives, and put $X:=\sum_{e\in D_i}w_e$.
Each preceding allocation round uses sampling probability $\le 1/(3n)$. These probabilities satisfy $q:=(1-1/3n)^{i-1}\ge 1-(i-1)/3n\ge 2/3$. Applying Lemma~\ref{lem:iterative-cylinder} to the weights $12w_e\in[0,1]$ with $\lambda\ge0$,
\[
    \mathbb{E}[e^{12\lambda X}]
    \le\exp\!\left((1-q) 12W(e^\lambda-1)\right)
    \le\exp\!\left(4W(e^\lambda-1)\right).
\]
Taking $\lambda=\ln(5/4)$ and applying Markov's inequality gives
\begin{align*}
    \Pr\!\left[\sum_{e\in R_i}w_e<\frac{7W}{12}\right]
    &=\Pr[X>5W/12]
    \le e^{-5\lambda W}\mathbb{E}[e^{12\lambda X}]\le\exp\!\left(-5W\ln(5/4)+W\right)
     <e^{-W/10}.
\end{align*}
For an additive agent, take $w_e=v_i(e)$ and use $W=v_i([m])\ge n$.
\end{proof}

    

\begin{lemma}
\label{lem:frac-alloc-simple}
    For any $i\in [n]$, if agent $i$ follows the safe strategy $\ws$, the agent receives $v_i(A_i)\ge \frac{1}{3n}v_i(R_i)-\frac{1}{12}$.
\end{lemma}
\begin{proofsketch}
    Observe that if agent $i$ follows $\ws$, for any $j\in [\ell_i-1]$, the value of any element in $B_{ij}$ is no smaller than the value of every element in $B_{i (j+1)}$. In particular, the value of the random element chosen in $A_{ij}$ is at least a $1/3n$ fraction of the total value of $B_{i(j+1)}$. This covers the values of all but the first bucket $B_{i1}$, but the total value of $B_{i1}$ is at most $3n$ times that of the most valuable element. Therefore the claim follows.
\end{proofsketch}

Putting the two lemmas together, we obtain the following theorem.
\begin{theorem}
\label{thm:warmup-additive}
If all agents have additive valuations and follow their respective
safe strategies $\ws$, then Algorithm~\ref{alg:add-simple} allocates to every agent
simultaneously a $1/12$-MMS with probability at least
$1-ne^{-n/10}$.
\end{theorem}

\begin{proof}
By Lemma~\ref{lem:warmup-surviving-weight} and a union bound, with
probability at least $1-ne^{-n/10}$, every agent $i$ has
$v_i(R_i)\ge7n/12\ge n/2$. On this event, Lemma~\ref{lem:frac-alloc-simple} gives $v_i(A_i)\ge\frac{v_i(R_i)}{3n}-\frac1{12} \ge\frac1{12}$ for all $i$.
\end{proof}
\subsection{Modifications needed for submodular agents}

We now briefly describe the changes needed in our algorithm and analysis to handle submodular values. We continue to operate under the assumption that $\MMS_i=1$, and $v_i(e)\le 1/12$, for all $i,e$. 

For the algorithm and the agent's strategy, a simple modification suffices: the algorithm allows each agent to specify the bucketing of items in Step~\ref{step:simple-bucketing} adaptively, with each bucket $B_{ij}$ chosen after the random selection of $A_{ij'}$ from $B_{ij'}$, $j'<j$, has been revealed. The agent chooses buckets greedily according to the marginal values of the elements, as formalized below.

\begin{definition}\label{def:safe-smallVal}
    {\bf (Safe Strategy $\ws$ for submodular agent $i$)} Define $\ell_i = \lceil |R_i|/3n \rceil$. For $j\in [\ell_i-1]$, set $B_{ij}$ to be the $3n$ elements in $R_i\setminus\cup_{j'<j} B_{ij'}$ that have the highest marginal values relative to the allocation so far, $\cup_{j'<j} A_{ij'}$. Break ties arbitrarily. Set $B_{i\ell_i} := R_i\setminus\cup_{j'<\ell_i} B_{ij'}$.
\end{definition}

Recall that for additive agents, we showed that with high probability $v_i(R_i)=\Omega(n)$ and the $1$-out-of-$n$ \MMS\ value of $R_i$ is $\Omega(1)$. 
Neither statement is necessarily true for submodular agents. Instead, we use a different potential function to track the value lost by the agent. Let $(P^i_1, \ldots, P^i_n)$ denote the $\MMS$-defining partition of items for agent $i$. For a subset $R$ of items, define the potential function $\Pot_i(R)$ as
\[\Pot_i(R) := \sum_{j\in [n]} v_i(R\cap P^i_j)\]
Observe that $\Pot_i([m]) \ge n$ by our assumption that $\MMS_i=1$. 

We can now restate the counterparts of Lemmas~\ref{lem:warmup-surviving-weight} and \ref{lem:frac-alloc-simple} in terms of $\Pot_i$. 

\begin{lemma}
\label{lem:enough-remains-submod-simple}
    For any $i\in [n]$, regardless of the partitioning strategies of agents $1, \cdots, i-1$, with probability at least $1-e^{-n/10}$, we have $\Pot_i(R_i)\ge 7n/12$.
\end{lemma}
\begin{proof}
    Consider each component $P^i_j$ of agent $i$'s MMS defining partition. Order the elements in $P^i_j$ in a greedy fashion by marginal value and assign to each a ``supporting value'' equaling the corresponding marginal value: let $g_1$ be the item with highest marginal value over $\emptyset$, let $g_2$ be the item with highest marginal value given $\{g_1\}$, and so on. Define $a_{g_q} := v_i(g_q \mid g_1,\dots,g_{q-1})$ for all $q \in [|P^i_j|]$. Then $\sum_{g\in P^i_j} a_{g} = v_i(P^i_j)$, while for any $R$, $\sum_{g\in P^i_j\cap R} a_{g} \le v_i(P^i_j\cap R)$. By submodularity, $0\le a_g\le v_i(g)\le1/12$ for every good $g$. Moreover, their total weight is $W:=\sum_{g\in[m]}a_g
      =\sum_{j=1}^n v_i(P_j^i)
      =\Pi_i([m])\ge n$.

    Applying the weighted-survival bound of Lemma~\ref{lem:warmup-surviving-weight} to these fixed weights, with probability at least $1-e^{-W/10}\ge 1-e^{-n/10}$, we obtain $\Pi_i(R_i)
    \ge\sum_{g\in R_i}a_g
    \ge 7n/12$.
\end{proof}

\begin{lemma}
\label{lem:frac-alloc-submod-simple}
    For any $i\in [n]$, if agent $i$ follows the safe strategy $\ws$, the agent receives $v_i(A_i)\ge \frac{1}{4n}\Pot_i(R_i)-\frac{1}{16}$. If $\Pot_i(R_i)\ge 7n/12$, then we obtain $v_i(A_i)\ge\frac{1}{12}$.
\end{lemma}
\begin{proofsketch}  
For every MMS bundle $P_h^i$, monotonicity and submodularity give
\[
    v_i(P_h^i\cap R_i)
    \le v_i(A_i)
       +\sum_{g\in(P_h^i\cap R_i)\setminus A_i}
            v_i(g\mid A_i).
\]
Since the MMS bundles partition the goods, summing over $h$ yields
\[
    \Pi_i(R_i)
    \le n v_i(A_i)
       +\sum_{g\in R_i\setminus A_i}v_i(g\mid A_i).
\]

For each sampled bucket $B_{ij}$, let $\Delta_j$ denote the
marginal value gained when its sampled item is added to the
allocation. Thus $\sum_{j<\ell_i}\Delta_j=v_i(A_i)$.
By greedy bucket selection and diminishing marginal values,
every unallocated item $g\in B_{i(j+1)}$ has
$v_i(g\mid A_i)\le\Delta_j$.
Each bucket, including the final unallocated bucket, has at most
$3n$ items. The unallocated items in the first bucket have total
marginal value at most $3n/12=n/4$, by the singleton-value bound.
Consequently,
\[
    \sum_{g\in R_i\setminus A_i}v_i(g\mid A_i)
    \le \frac n4
       +3n\sum_{j=1}^{\ell_i-1}\Delta_j
    =\frac n4+3n v_i(A_i).
\]
Combining the two bounds gives
$\Pi_i(R_i)\le4n v_i(A_i)+n/4$,
which rearranges to the claimed inequality.
\end{proofsketch}


Putting the two lemmas together, we obtain the following theorem.

\begin{theorem}\label{thm:sub-smallVal}
    If each agent $i\in [n]$ has a submodular valuation, where value of every single item is less than $\frac{1}{12}$, and follows the safe strategies as in \Cref{def:safe-smallVal} to construct $B_{i1},\dots, B_{i\ell_i}$ adaptively within Algorithm~\ref{alg:add-simple}, then it allocates to every agent simultaneously an $1/12$ approximation to MMS with probability at least $1-ne^{-n/10}$.
\end{theorem}

\section{Main Result: Constant MMS for Submodular Valuations}
\label{sec:submod}
We will now present our main result: an $(\alpha, \beta)$ approximation for the online MMS problem where agents have submodular valuations with both $\alpha \in \Theta(1)$ and $\beta \in \Theta(1)$. Our algorithm and analysis follows the ideas discussed in Section~\ref{sec:additive}. But removing the ``no high value items'' assumption requires significant extra effort.

The challenge with high value items is that a few items can capture a large fraction of the agent's potential function value, so randomizing the agent's allocation over a large set no longer guarantees the agent enough value. As an extreme example, suppose that an additive agent values $n$ items at $1$ and the remaining at $0$. Then, even if all the good items are present when the agent arrives, and the agent follows the safe strategy $\ws[]$ in Algorithm~\ref{alg:add-simple}, the agent has no more than a $1/3$ chance of getting a value $>0$. This issue can be fixed through a ``single-item reduction'': we allow the agent to pick any one large-value item with certainty. The key observation is that this should not hurt others' MMS value. If all of the agents with high value items arrive first, and each takes one item, for the remaining, say $n'<n$, agents the $1$-out-of-$n'$-MMS value {out of the remaining set of items} is still at least $1$. Running Algorithm~\ref{alg:add-simple} with a revised value $n'$ for the number of agents, should then provide our desired guarantees. The issue is that the high-value agents might be interspersed with low-value ones. So we need to adjust our allocation probabilities for the low-value agents in a manner that reflects the changing $n'$.

These ideas are formalized in Algorithm~\ref{alg:submod-rr} below. Rounds with high-value agents are called ``special'' and the variable $t$ keeps track of the number of such rounds seen by the algorithm so far. 



\begin{algorithm}[H]
\caption{Single-or-Sample Allocation Algorithm}
\label{alg:submod-rr}
\begin{algorithmic}[1]
\Require Number of agents $n$, ground set $[m]$, parameter $k$
\State $R_1 \gets [m]$; $t \gets 0$
\For{$i = 1,2,\dots,n$}
    \State Agent $i$ observes the current remaining set $R_i$
    \State Agent $i$ chooses one of the following two actions:
    \If{agent $i$ chooses a special round}
        \State Agent $i$ selects an item $e \in R_i$
        \State Allocate $e$ to agent $i$
        \State $R_{i+1} \gets R_i \setminus \{e\}$; $t \gets t+1$
    \Else
        \State Initialize $G=R_i$ and $j=1$ \label{step:rr-start}
        \While{$|G|\ge k(n-t)$} \label{step:while}
            \State Agent $i$ specifies bucket $B_{ij}\subset G$ with $|B_{ij}|=k(n-t)$ \label{step:adap-buckets}
            \State Pick a uniformly random item $e_{ij}$ from $B_{ij}$
            \State $G\gets G\setminus B_{ij}$, $j \gets j+1$
        \EndWhile
        \State Allocate $A_i=\cup_j\{e_{ij}\}$ to agent $i$
        \State $R_{i+1} \gets R_i \setminus A_i$ \label{step:rr-end}
    \EndIf
\EndFor
\end{algorithmic}
\end{algorithm}

\subsection{Safe strategy and \MMS\ guarantee}

We next define a ``safe strategy'' for agent $i$, denoted $\ws$. By following this strategy, a submodular agent can guarantee a constant fraction of her \MMS\ value with high probability, regardless of the actions of the other agents. 
{For a fixed parameter $c$ to be specified later, the following safe strategy targets value $\tau=1/ck$ for each agent.}
\begin{definition}\label{def:safe}
    {\bf (Safe Strategy $\ws$ for agent $i$ in Algorithm~\ref{alg:submod-rr})}
    If $R_i$ contains an item $e$ with $v_i(e)\ge \tau$, then agent $i$ chooses a special round and selects an item of maximum value, $e\in \argmax_{g\in R_i} v_i(g)$.    Otherwise, the agent participates in the randomized iterative allocation in Steps~\ref{step:rr-start}--\ref{step:rr-end} and receives the resulting allocation $A_i$.

    In round $j$ of the while loop (Step~\ref{step:while}), let $A_i^{j-1}:=\{e_{ij'}:j'<j\}$ denote the items allocated to agent $i$ in earlier rounds. The marginal value of an item $e\in G$ is
    \[
        v_i(e\mid A_i^{j-1})
        :=
        v_i(A_i^{j-1}\cup\{e\})-v_i(A_i^{j-1}).
    \]
    The agent chooses $B_{ij}$ to consist of the $k(n-t)$ items in $G$ with the largest marginal values.
\end{definition}

\begin{theorem}
\label{thm:submod-rr-main}
    For every instance $\I = ([n],[m],(v_i)_{i\in[n]})$, setting $k=4$ and $c=6$ ensures the existence of failure probabilities $q_i$, $i\in[n]$, such that:
    \begin{itemize}
        \item For every agent $i\in[n]$ with a submodular valuation function, following the safe strategy $\ws$ in \Cref{alg:submod-rr} guarantees an allocation of value at least $\frac{1}{24}\,\MMS_i$ with probability at least $1-q_i$, regardless of the actions of the other agents.
        \item The failure probabilities satisfy $\sum_{i\in[n]} q_i \le 0.577$.
    \end{itemize}
\end{theorem}

\begin{corollary}
    If all agents $i\in[n]$ have submodular valuation functions and follow their respective safe strategies $\ws$, then \Cref{alg:submod-rr} is $\left(\frac{1}{24},\,0.423\right)$-competitive.
\end{corollary}

\paragraph{Incentive properties.} 

Note that Algorithm \ref{alg:submod-rr} requires only $O(m\log m)$ bits of information from each agent---one bit to decide whether to choose special round or non-special round followed by ordinal preferences over the goods. In particular, it is independent of the precision required to express the valuation functions of the agents. Further, 
\Cref{thm:submod-rr-main} implies a strong incentive guarantee, namely following the {\em safe strategy} ensures $1/24$-MMS to an agent no matter how strategic other agents are. Finally, being truthful is itself the safe strategy for each agent,  making it ``simple'' for each agent to implement.

\subsection{Analysis of a single agent}
\label{sec:submod-anal}

Fix an agent $i$, and normalize her MMS value to one.  We analyze the
safe strategy of agent $i$ against arbitrary behavior of the agents that
arrive before her.  
If an item of
value at least $\tau=1/ck$ is available when agent $i$ arrives, she takes such
an item in a special round; otherwise she participates in the adaptive
bucketing procedure of the Single-or-Sample Algorithm.

\subsubsection{Greedy bucketing (Counterpart of Lemma~\ref{lem:frac-alloc-submod-simple})}

We first prove that if enough bundles of the agent retain enough total value, then the agent deterministically receives a large value allocation.

\begin{lemma}[Greedy bucketing guarantee]
\label{lem:submod-greedy-bucketing}
Suppose agent $i$ arrives after $t$ special rounds, and let
$N:=n-t$.  Let $Q_1,\ldots,Q_N$ be any pairwise disjoint sets of goods,
and define
\[
    \Pot:=\sum_{h=1}^N v_i(Q_h\cap R_i).
\]
If agent $i$ follows the safe strategy $\ws$, then
\[
    v_i(A_i)
    \ge
    \min\left\{
        \tau,
        \frac{1}{k+1}\left(\frac{\Pot}{N}-k\tau\right)
    \right\}.
\]
In particular, if
$\Pot\ge (2k+1)\tau N$, then $v_i(A_i)\ge\tau$. 
\end{lemma}

\begin{proof}
If $R_i$ contains an item of value at least $\tau$, the safe strategy
chooses a special round and obtains value at least $\tau$.  We may
therefore assume that $v_i(e)<\tau$ for all $e\in R_i$. The agent then runs the adaptive bucketing procedure with bucket size $kN$.

If $|R_i| < kN$, the procedure forms no full bucket, so $A_i = \emptyset$ and
$v_i(A_i) = 0$. In this case the bound holds trivially. Indeed, since the sets
$Q_1, \dots, Q_N$ are pairwise disjoint and every $e \in R_i$ satisfies
$v_i(e) < \tau$,
\[
  \Pi \;=\; \sum_{h=1}^{N} v_i(Q_h \cap R_i)
      \;\le\; \sum_{g \in R_i} v_i(g)
      \;<\; |R_i|\,\tau
      \;<\; kN\tau ,
\]
so $\frac{1}{k+1}\bigl(\frac{\Pi}{N} - k\tau\bigr) < 0 \le v_i(A_i)$. In
particular, this case cannot occur when $\Pi \ge (2k+1)\tau N$. We therefore
assume from now on that at least one full bucket is formed, i.e., $\ell \ge 1$.

Let $B_1,\ldots,B_\ell$ be the full buckets produced by the procedure, let $e_s$ be the item selected from $B_s$, and let
$E_s:=\{e_1,\ldots,e_s\}$, with $E_0=\varnothing$. Note that $A_i = E_{\ell}$. Define $\Delta_s:=v_i(e_s\mid E_{s-1})$, so that $\sum_{s=1}^{\ell}\Delta_s= v_i(A_i)$.  Let $L$ be the final
set of items that do not form a full bucket.

We charge every item in $R_i\setminus A_i$ as follows.  If
$g\in B_s\setminus\{e_s\}$, set $c_g:=v_i(g\mid E_{s-1})$, and if $g\in L$, set $c_g:=v_i(g\mid E_\ell)$.  Let
$C:=\sum_{g\in R_i\setminus A_i}c_g$.

For every $h\in[N]$, monotonicity and submodularity imply
\begin{align*}
    v_i(Q_h\cap R_i)-v_i(A_i)
    &\le v_i((Q_h\cap R_i)\cup A_i)-v_i(A_i) \\
    &\le \sum_{g\in (Q_h\cap R_i)\setminus A_i}v_i(g\mid A_i) \le \sum_{g\in (Q_h\cap R_i)\setminus A_i}c_g.
\end{align*}

Since the $Q_h$'s are disjoint, summing over $h$ gives
\begin{equation}
\label{eq:charge-lower}
    \Pot-Nv_i(A_i)\le C.
\end{equation}

We next upper bound $C$.  Since $v_i(e)\le\tau$ for all items $e$, the total charge
of the unchosen items in the first bucket is at most $kN\tau$.  For
$s\ge2$, every $g\in B_s$ was available when $B_{s-1}$ was formed.
Since $B_{s-1}$ consists of the $kN$ items of largest marginal value at
that point,
\[
    c_g
    =v_i(g\mid E_{s-1})
    \le v_i(g\mid E_{s-2})
    \le v_i(e_{s-1}\mid E_{s-2})
    =\Delta_{s-1}.
\]
Likewise, every $g\in L$ was available when $B_\ell$ was formed, so
$c_g\le\Delta_\ell$.  Hence
\begin{equation}
\label{eq:charge-upper}
    C\le kN\tau+kN\sum_{s=1}^{\ell}\Delta_s
    =kN(\tau+v_i(A_i)).
\end{equation}

Combining~\eqref{eq:charge-lower} and~\eqref{eq:charge-upper}, and rearranging, we get,
\[
    v_i(A_i)
    \ge
    \frac{1}{k+1}\left(\frac{\Pot}{N}-k\tau\right).
\]
Together with the special-round guarantee this proves the claim.
\end{proof}

\subsubsection{The deficit: a weighted loss function}

We now track the loss in value caused by agents preceding \(i\). Rather than counting the number of previously allocated items, we track the loss through supporting prices of the items defined by the agent's submodular valuation function.

Special rounds introduce an additional complication. If \(t_i\) special rounds occur before agent \(i\) arrives, we would ideally restrict attention to \(n-t_i\) of agent \(i\)'s \MMS\ bundles, excluding one bundle touched by each special round. However, \(t_i\), as well as the identities of the bundles touched by special rounds, may depend on the algorithm's previous random choices, so we cannot fix this collection of bundles in advance and apply concentration to it directly. We circumvent this issue by maintaining a set of {\em active} bundles; discarding one active \MMS\ bundle after each special round; and defining a deficit that measures the potential lost from the active bundles relative to their {\em nominal potential}, namely the potential they would have if each active bundle retained its full value $1$. More precisely, the deficit is defined to be the number of active bundles minus the sum of the supporting prices of the unallocated items in the active bundles. We now provide the formal definitions.

\paragraph{Supporting prices.}
Let $(P_1,\ldots,P_n)$ be an \MMS-defining partition for agent $i$.
For each bundle $P_h$, fix an arbitrary ordering
$g_{h,1},\ldots,g_{h,\ell_h}$ and define the chain marginal
$a_{g_{h,s}}
    :=v_i\bigl(g_{h,s}\mid \{g_{h,1},\ldots,g_{h,s-1}\}\bigr)$.
We normalize these marginals within each \MMS\ bundle and define the
\emph{supporting prices} $w_e:=\frac{a_e}{v_i(P_h)}$ for $e\in P_h$. These prices have the following properties:

\begin{equation*}
    \sum_{e\in P_h}w_e=1,
    \qquad
    0\le w_e\le 1,
    \qquad
    \sum_{e\in Q}w_e\le v_i(Q)
    \quad\text{for every }Q\subseteq P_h.
\end{equation*}

\paragraph{Active \MMS\ bundles and the deficit.} Initially every \MMS\ bundle is \emph{active}. After each non-special round the active collection is unchanged.  After a special round, if the
item allocated in that round belongs to a currently active bundle, we
deactivate that bundle.  Otherwise we deactivate an arbitrary currently
active bundle.  These deactivations are only part of the analysis and do
not alter the algorithm.

For $j=0,1,\ldots,i-1$, let $t_j$ be the number of special rounds among
the first $j$ rounds, let \(N_j:=n-t_j\),
and let $\mathcal B_j\subseteq[n]$ denote the indices of the active \MMS\
bundles after these $j$ rounds.  Thus $N_j=|\mathcal B_j|$ is the nominal potential.  Recall that
$R_{j+1}$ is the set of goods remaining after the first $j$ agents have
been served.  Define the surviving supporting weight
\begin{equation}
    W_j
    :=
    \sum_{h\in\mathcal B_j}
    \sum_{e\in P_h\cap R_{j+1}} w_e, \qquad \text{and the \emph{deficit},} \qquad D_j:=N_j-W_j. \label{eq:deficit-definition}
\end{equation}
Initially $W_0=N_0=n$, and hence $D_0=0$.

The reason for deactivating one bundle in every special round is that a
special round can then never increase the deficit.  Indeed, suppose the
surviving supporting weight of the bundle deactivated in a special round
is $b$.  Since each MMS bundle has total supporting weight one, $b\le1$.
The special round changes $N_j$ to $N_j-1$ and $W_j$ to $W_j-b$, and so $D_{j+1}=D_j-(1-b)\le D_j$.
If the special item lies in an active bundle, that is precisely the
bundle we deactivate; if it lies outside the active bundles, its
allocation removes no active supporting weight and we may deactivate any
active bundle.



\subsubsection{An exponential-moment bound for the deficit}

Let $r_j:=j-t_j$ denote the number of non-special rounds among the first
$j$ rounds.  For $\lambda\ge0$, define $\psi(\lambda):=(e^\lambda-1)/k$.
We next prove the exponential-moment estimate that will replace the
fixed-set concentration argument.

\begin{lemma}[Deficit exponential moment]
\label{lem:deficit-mgf}
For every $j\in\{0,1,\ldots,i-1\}$ and every $\lambda\ge0$,
\begin{equation}
\label{eq:deficit-mgf}
    \mathbb E\left[
        \exp\bigl(\lambda D_j-\psi(\lambda)r_j\bigr)
    \right]
    \le1.
\end{equation}
The expectation is over the random choices of the algorithm in the first
$j$ rounds.  The bound holds even when the preceding agents choose
special versus non-special rounds, and choose their buckets, as arbitrary
functions of the history so far.
\end{lemma}

\begin{proof}
Fix $\lambda\ge0$. We will apply Lemma~\ref{lem:conditional-mgf-composition} to the process $\{D_j\}$. Let $\mathcal{F}_j$ denote the history through
the beginning of round $j+1$, including the arriving agent's choice
of whether the round is special or non-special. Define the
$\mathcal{F}_j$-measurable charge
\[
    b_{j+1}:=
    \psi(\lambda)\,
    \mathbf{1}_{\{\text{round }j+1\text{ is non-special}\}},
    \qquad j=0,\ldots,i-2.
\]
We verify the one-round bound required by
Lemma~\ref{lem:conditional-mgf-composition}.

If round $j+1$ is special, then $D_{j+1}\le D_j$, as established
above. Since $\lambda\ge0$ and $b_{j+1}=0$,
\[
    \mathbb{E}\left[
        e^{\lambda(D_{j+1}-D_j)}
        \,\middle|\,\mathcal{F}_j
    \right]
    \le 1=e^{b_{j+1}}.
\]

Otherwise, give each available item $e\in R_{j+1}$ weight
$x_e=w_e$ if its MMS bundle is currently active, and $x_e=0$
otherwise. These weights are determined by $\mathcal{F}_j$ and
satisfy
\[
    0\le x_e\le1,
    \qquad
    \sum_{e\in R_{j+1}}x_e=W_j\le N_j.
\]
Let $A$ be the allocation in this round. Since no bundle is
deactivated, $D_{j+1}-D_j=\sum_{e\in A}x_e$.
The adaptive bucketing procedure satisfies upper-cylinder domination
with parameter $p=1/(kN_j)$; leaving the final incomplete bucket
unallocated preserves this bound. Applying Lemma~\ref{lem:weighted-mgf} gives
\[
    \mathbb{E}\left[
        e^{\lambda(D_{j+1}-D_j)}
        \,\middle|\,\mathcal{F}_j
    \right]
    \le
    \exp\left(
        \frac{W_j}{kN_j}(e^\lambda-1)
    \right)
    \le
    e^{\psi(\lambda)}
    =
    e^{b_{j+1}}.
\]

Thus the required conditional MGF bound holds in every round, and Lemma~\ref{lem:conditional-mgf-composition} along with $\sum_{\ell=1}^j b_\ell=\psi(\lambda)r_j$ implies the required bound.
\end{proof}

\subsubsection{Failure probability for one agent (Counterpart of Lemma~\ref{lem:enough-remains-submod-simple})}
We now convert the exponential-moment bound into a failure probability. Let $t_i:=t_{i-1}$ be the number of special rounds before agent $i$
arrives, $N_i:=n-t_i$ be the number of active MMS bundles at the arrival of agent
$i$, and $r_i:=i-1-t_i$ be the number of preceding non-special rounds.

Recall that agent $i$ succeeds if $v_i(A_i)\ge\tau$, and by Lemma~\ref{lem:submod-greedy-bucketing}, it suffices to ensure that the total residual value of the $N_i$ active \MMS\ bundles is at least $(2k+1)\tau N_i$. Furthermore, with $(w_e)_{e\in G}$ being the supporting prices defined previously, and $\mathcal B_{i-1}$ being the indices of agent $i$'s active bundles, we have 
\[\sum_{h\in\mathcal B_{i-1}} v_i(P_h\cap R_i) \ge \sum_{h\in\mathcal B_{i-1}}\sum_{e\in P_h\cap R_i} w_e = N_i-D\]
where $D$ is the deficit when agent $i$ arrives. So, in order for the agent to succeed, it suffices to ensure $D\le (1-(2k+1)\tau) N_i$. We can now apply the exponential moment bound along with Markov's inequality at an appropriate value of $\lambda$ to obtain the following bound.


\begin{lemma}[Single-agent failure probability]
\label{lem:submod-single-agent-failure}
Let $\gamma = 1-(2k+1)\tau$.
Suppose $\lambda>0$ satisfies
\begin{equation}
\label{eq:lambda-condition}
    \psi(\lambda)=\frac{e^\lambda-1}{k}
    \le \lambda\gamma.
\end{equation}
Then the failure probability of agent $i$ is at most
\begin{equation}
\label{eq:qi-bound}
    q_i:=\Pr[v_i(A_i)<\tau]
    \le
    \exp\bigl(-\lambda\gamma(n-i+1)\bigr).
\end{equation}
\end{lemma}

\begin{proof}
If agent $i$ fails, then we must have $D>(1-(2k+1)\tau)N_i=\gamma N_i$. Therefore, using $N_i=n-i+1+r_i$, we get
\begin{align*}
    \lambda D-\psi(\lambda)r_i
    &>\lambda\gamma N_i-\psi(\lambda)r_i\\
    &=\lambda\gamma(n-i+1)
      +\bigl(\lambda\gamma-\psi(\lambda)\bigr)r_i\\
    &\ge\lambda\gamma(n-i+1),
\end{align*}
where the final inequality uses~\eqref{eq:lambda-condition}.  

Applying Markov's inequality and Lemma~\ref{lem:deficit-mgf} with
$j=i-1$ yields
\[
    q_i \le \Pr\Big[\lambda D-\psi(\lambda)r_i \ge \lambda\gamma(n-i+1)\Big]
    \le
    e^{-\lambda\gamma(n-i+1)}
    \mathbb E\left[e^{\lambda D-\psi(\lambda)r_i}\right]
    \le
    e^{-\lambda\gamma(n-i+1)}.
\]
\end{proof}

\subsection{Bounding the total failure probability (Proof of Theorem~\ref{thm:submod-rr-main})}
\label{sec:submod-anal-2}
We now sum the single-agent bounds.  For any parameters $c,k,\lambda$
satisfying $\gamma=1-(2k+1)\tau>0$ and~\eqref{eq:lambda-condition}, a union bound and
Lemma~\ref{lem:submod-single-agent-failure} give
\begin{align*}
\Pr[\text{some agent fails}]
&\le \sum_{i=1}^n q_i 
\le \sum_{i=1}^n e^{-\lambda\gamma (n-i+1)} 
< \frac{1}{e^{\lambda\gamma}-1}.
\label{eq:total-failure-general}
\end{align*}
It is easy to verify that all constraints hold for the choices
\[
    k=4,
    \qquad
    c=6,
    \qquad
    \tau=\frac1{24},
    \qquad
    \gamma=1-9\tau=\frac58,
    \qquad
    \lambda=\ln 5.
\]
With this choice of parameters, the agent's \MMS\ approximation guarantee is $\tau=1/24$, and the total failure probability can be bounded as,
\[
    \sum_{i=1}^n q_i
    <
    \frac{1}{5^{5/8}-1}
    <0.577.
\]
This proves Theorem~\ref{thm:submod-rr-main}.
The constants are not optimized.

\section{Additional Properties of the Single-or-Sample Algorithm}\label{sec:additional}
In this section, we show that small modifications to the Single-or-Sample Algorithm yield additional guarantees. Throughout, we use the same parameter values as in \Cref{sec:submod}. The key observation underlying these extensions is that skipping a small number of agents substantially improves the probability of obtaining a good allocation for all remaining agents.


\begin{proposition}\label{prop:skip-failure-prob-single-agent}
    If Algorithm \ref{alg:submod-rr} is modified so that a fixed set $\A_{s}$ of agents is not allocated any goods by the algorithm, the probability of failure of each of the remaining agents is at most $q_i \leq e^{-\lambda \gamma (1+ |A_s|)}$. 
\end{proposition}
\begin{proof}
    We call any $i \in \A_s$ a skipped agent and other agents as unskipped agents. For an unskipped agent $i$, let $r_i$ be the total number of non-special rounds and $t_i$ be the total number of special rounds that happen before $i$ arrives. If there are any rounds previously that are skipped by the algorithm, we don't count them towards either $r_i$ or $t_i$. Additionally, we also don't change the deficit after any such rounds. Therefore, the exponential bound on deficit from \Cref{lem:deficit-mgf} continues to hold. In the proof of \Cref{lem:submod-single-agent-failure}, we get the bound on deficit to be $D_{i-1}> \gamma N_i = \gamma (n-t_i)$. If agent $i$ is the $i_{\ell}^{th}$ unskipped agent, then $r_i + t_i = i_{\ell} - 1$. Therefore, $D_{i-1} > \gamma(n + r_i - i_\ell + 1)$. Then, in the proof of \cref{lem:submod-single-agent-failure}, we only need to replace $i$ with $i_\ell$ in the final failure probability. Therefore, we get $q_i \leq e^{-\lambda \gamma (n-i_\ell + 1)}$. Since we are ignoring $|\A_s|$ agents, $i_{\ell} \leq n - |\A_s|$. Therefore, $q_i \leq e^{-\lambda \gamma (1+ |A_s|)}$ which proves the proposition.
\end{proof}

\begin{lemma}\label{thm:add-fixed-set-to-skip}
Let $\A_{\mathrm{skip}}$ be any fixed set of $\lceil 2\ln n\rceil$ agents. Consider the following modification of the Single-or-Sample Algorithm~\ref{alg:submod-rr}: agents in $\A_{\mathrm{skip}}$ are skipped and receive no items, while all remaining agents participate in the algorithm as usual. Then, with probability at least $1-\frac{1}{n}$, every agent $i\notin \A_{\mathrm{skip}}$ receives value at least $\frac{1}{24}\MMS_i$.
\end{lemma}
\begin{proof}
Renumber the agents in $[n]\setminus \A_{\mathrm{skip}}$ as $1,\ldots,n-\lceil 2\ln n\rceil$ in order of arrival. By \Cref{lem:submod-single-agent-failure}, agent $i$ receives value at least $\frac{1}{24}\MMS$ except with probability
$$
q_i\le \exp\!\bigl(-\lambda\gamma(n-i+1)\bigr)
\le \exp\!\bigl(-\lambda\gamma(1+|\A_{\mathrm{skip}}|)\bigr)
\le e^{-\frac58\ln 5\,(1+2\ln n)}
\le \frac{1}{n^2},
$$
where we use the parameter values from \Cref{sec:submod}. A union bound over the at most $n$ unskipped agents therefore gives total failure probability at most $1/n$.
\end{proof}
We now derive two consequences. The first improves the guarantee for each individual agent: every agent receives a constant fraction of her \MMS\ with probability $1-o(1)$. In contrast, under the Single-or-Sample Algorithm, the last agent has a constant failure probability. The modification below reduces this to $o(1)$ by slightly increasing the failure probabilities of the other agents, at the cost of a larger probability of the overall failure event, namely that {\em some} agent receives a low-value allocation.

\begin{corollary}
    For submodular valuations, Algorithm \ref{alg:submod-rr} can be modified so that every agent $i \in [n]$ receives a $\frac{1}{24}$-MMS allocation with probability at least $1-o(1)$.
\end{corollary}
\begin{proof}
    Choose $\A_{\mathrm{skip}}$ uniformly at random among all sets of $\lceil 2\ln n\rceil$ agents, and skip these agents while running Algorithm~\ref{alg:submod-rr} unchanged on the rest. By \Cref{thm:add-fixed-set-to-skip}, every unskipped agent receives at least $\frac{1}{24}$-\MMS\ with probability at least $1-\frac1n$. Thus, for any agent, the failure probability is at most
$$
\Pr[i\in \A_{\mathrm{skip}}]+\frac1n
\le \frac{\lceil 2\ln n\rceil}{n}+\frac1n
=O\!\left(\frac{\log n}{n}\right).
$$
\end{proof}
The second consequence is that if we relax the notion of \MMS\ to $1$-out-of-$N$-\MMS\ instead of $1$-out-of-$n$-\MMS\ for $N$ slightly larger than $n$, then the probability of the overall failure event, namely that some agent receives a low-value allocation, can be made polynomially small.

\begin{corollary}
    Let $N$ satisfy $N\ge n+\lceil 2\ln N\rceil$. For submodular valuations, Algorithm \ref{alg:submod-rr} can be modified so that all agents receive a $\frac{1}{24}$-$1$-out-of-$N$-\MMS\ simultaneously with probability $1-O\left(\frac{1}{n}\right)$. 
\end{corollary}
\begin{proof}
    Consider adding $N-n$ ``dummy'' agents with $0$ valuations to
    Algorithm \ref{alg:submod-rr}, and let $\A_{\mathrm{skip}}$ denote the set of dummy agents. Then \Cref{thm:add-fixed-set-to-skip} implies that all of the real agents are simultaneously satisfied with probability at least $1-1/N$.
\end{proof}
\section{Main result: $O(1/\log n)$-MMS for XOS Valuations}
\label{sec:xos-lower-bound}

Our positive results give a constant MMS approximation with constant probability
for submodular valuations. In this section, we show that this guarantee is impossible
for XOS valuations, even when we are aiming for a $(2/\log\log{n})$-MMS approximation for \emph{binary-XOS} valuations: maxima of additive
clauses whose item values are all zero or one. In \Cref{app:xos-lb} we present a construction and the full proof of our lower bound result for XOS valuations stated in the following theorem.

\begin{restatable}{theorem}{xoslb}
\label{thm:xos-lower-bound}
No randomized online allocation
algorithm can guarantee a $\frac{2}{\log\log{n}}$-MMS allocation with probability $\Theta(1)$ 
against an oblivious adversary, even for simple binary-XOS valuations with no large items (where $\max_{i,g} v_i({g})/\MMS_i \le n^{-2}$).
\end{restatable}

The rest of this section is devoted to our upper bound result. We show that we can obtain an $O(\frac{1}{\log{n}})$-MMS guarantee with the probability of success being $1-\Theta(\frac{1}{\log{n}})$. \\

\subsection{Tight Logarithmic MMS Approximation for XOS valuations}\label{sec:xos-upper-bound}
We now provide an algorithm which is $(\Theta(\frac{1}{\log{n}}),1-O(\frac{1}{\log{n}}))$-MMS competitive for sufficiently large $n$. Like the Single-or-Sample algorithm, it allocates either a single high-valued good or a bundle obtained through random sampling. The key difference lies in the sampling rule. For submodular valuations, adaptive greedy bucketing guarantees each agent a constant fraction of her MMS whenever a sufficiently high value remains across her MMS-defining bundles. This guarantee, however, does not extend to XOS valuations.

For example, consider an instance with \(n\) agents and \(n^2\) goods. Fix an agent \(i\), and partition the goods into \(n\) sets \(P_1,\ldots,P_n\), each of size \(n\). Let $v_i(S)=\max_{j\in[n]}|S\cap P_j|.$ Then \(\mathrm{MMS}_i=n\), while every good has singleton value \(1\). Consider a greedy bucketing procedure with buckets of size \(n\), where each bucket consists of goods with the highest marginal values relative to the agent’s current allocation. Initially, all goods have equal marginal value, so a valid tie-breaking rule can select \(P_1\) as the first bucket. After allocating one uniformly random good from \(P_1\) and setting aside the rest of that bucket, every good in an unprocessed part has zero marginal value. The procedure can therefore continue by selecting each remaining part as a bucket. The resulting allocation contains exactly one good from each part and consequently has value only \(1\), or a \(1/n\) fraction of the agent’s MMS, even though all goods were initially available.

This example motivates a sampling process that has more correlation among the selected goods. We achieve this through a two-stage sampling process. First, we consider the agent’s MMS-defining partition and identify the bundles whose remaining value is at least a constant fraction of her MMS. Provided sufficiently many such bundles exist, we select one uniformly at random. We then sample only from the remaining goods in that bundle, using a procedure that guarantees the agent a \(\Theta(1/\log n)\) fraction of her MMS. Algorithm \ref{alg:XOS} formalizes this approach. Note that the algorithm allocates high valued goods using a fixed order. We require this for a technical reason that will be clear in the analysis.

\begin{algorithm}[ht!]
\caption{Random Bundle Allocation Process for XOS}
\label{alg:XOS}
\begin{algorithmic}[1]
\Require Number of agents $n$, ground set $[m]$, parameter $\tau \leq \frac{1}{4\max\{1500, \log n\}}$
\State $R_1 \gets [m]$; $t \gets 0$
\State Fix an arbitrary ordering of the goods 
\For{$i = 1,2,\dots,n$}
    \If{there's an item $e\in R_i$ with $v_i(\{e\})\ge \tau$}
        \State Allocate the first available good $e^*$ such that $v_i(\{e^*\}) \ge \tau$ from the global ordering
        \State $R_{i+1} \gets R_i \setminus \{e^*\}$;  $t \gets t+1$
    \Else
        \State Suppose $i$'s MMS partition of the ground set $[m]$ is $P^i_1,\dots,P^i_n$
        \State Let $\mathcal L$ denotes the set of indices $\ell$ with $v_i(P^i_\ell\cap R_{i})\geq \frac{1}{2}$
        \If{$|\mathcal L|<\frac{n-t}{2}$} 
            \State report failure
        \Else
            \State Pick a uniformly random index $\ell$ from $\mathcal L$
            \State Initialize $G=P_\ell\cap R_i$; $j \gets 1$; $k \gets \lfloor 1/(4\tau) \rfloor$
            \State Order goods in $G = (e_1, \ldots, e_{|G|})$  according to its supporting additive clause
            \For{$j \in \big\lfloor \frac{|G|}{k} \big\rfloor$} \label{step:xos-sample}
                \State Select one good $e_{ij}$ uniformly at random from $(e_{(j-1)k+1}, \ldots, e_{jk})$
            \EndFor
        \State Allocate $A_i=\cup_j\{e_{ij}\}$ to agent $i$; $R_{i+1} \gets R_i \setminus A_i$ 
        \EndIf
    \EndIf
\EndFor
\end{algorithmic}
\end{algorithm}
Our main result for Algorithm \ref{alg:XOS} is:
\begin{theorem}\label{thm:xos-alg}
    If all agents $i\in[n]$ have XOS valuation functions then for any $\tau \leq \frac{1}{4 \max \{1500, \log n\}}$ \Cref{alg:XOS} is $(\tau, 1-1000 \tau)$-competitive. 
\end{theorem}


\subsubsection{Analysis}
\label{sec:xos-analysis}

We normalize every agent's MMS value to one and fix, independently of
all random choices of the algorithm, the MMS partition
$(P_1^i,\ldots,P_n^i)$ used by each agent $i$.
For any given $\tau$, define
\[
    x := \frac{1}{4\tau},
    \qquad
    k:=\left\lfloor x \right\rfloor,
    \qquad
    p:=\frac1k.
\]
We first show that the algorithm either reports failure or allocates each agent the
guaranteed approximation to MMS. In particular, every allocation made before a reported failure has the required value.

\begin{lemma}[Value of a successful allocation]
\label{lem:xos-allocation-value}
Every agent served by Algorithm~3 before it reports failure receives
value at least $\tau$. Consequently, if the algorithm does not report
failure, all agents receive a $\tau$-MMS allocation.
\end{lemma}

\begin{proof}
An agent served in a special round receives an item of value at
least $\tau$. Otherwise, the algorithm selects a residual bundle
$G=P_\ell^i\cap R_i$ with $v_i(G)\ge1/2$.
Let $a$ be the nonnegative supporting additive clause used to order
$G$, so that $a(G)=v_i(G)$ and $a(U)\le v_i(U)$ for every $U$.
Every good $e\in G$ satisfies $a(\{e\})\le v_i(\{e\})<\tau$.
Since $k\tau\le{1/4}<1/2$, there is at least one full block.

Let $B_1,\ldots,B_\ell$ be the full blocks in non-increasing order
of supporting value, and let $B_{\ell+1}$ be the possibly empty
incomplete block. The item sampled from $B_j$ has supporting value
at least that of every item in $B_{j+1}$. Thus its supporting value
is at least $a(B_{j+1})/k$. Summing over the full blocks gives
\[
    v_i(A_i)
    \ge a(A_i)
    \ge \frac{a(G)-a(B_1)}k
    \ge \frac{1/2-k\tau}{k}
    \ge \tau,
\]
where the final inequality follows from $k\tau\le {1/4}$.
\end{proof}

It remains to bound the probability of a reported failure. As in the submodular setting, the difficulty is that allocations in special rounds are not random, and may indeed be adversarial. We show that nevertheless, special rounds have a limited effect on non-special rounds. We first focus exclusively on the weight loss in non-special rounds, and then incorporate the effect of special rounds using a shadow allocation process in our overall argument.

For the analysis, we will stop the actual algorithm at its first reported failure. Let $\nspc_i$ and $\spc_i$ be the goods allocated in non-special and
special rounds, respectively, before round $i$ of this stopped process.
These sets remain unchanged after stopping. Let $\nspc:=
\nspc_{n+1}$, and $t_i:=|\spc_i|$. For every agent $i$ reached by the algorithm,
$R_i=[m]\setminus(\nspc_i\cup \spc_i)$.

\paragraph{Weighted loss in non-special rounds.}
The next lemma combines the weighted MGF bound from Section~\ref{sec:tools} with
an averaging step over the randomly selected bundle.

\begin{lemma}[Weighted loss under random-bundle sampling]
\label{lem:xos-weighted-loss}
Fix weights $x_e\in[0,1]$ for all goods, independently of the
algorithm's random choices. Define
\[
    W:=\sum_{e\in[m]}x_e,
    \qquad
    X:=\sum_{e\in \nspc}x_e.
\]
Then, for every $\lambda\ge0$,
\begin{equation}
\label{eq:xos-random-bundle-mgf}
    \mathbb{E}[e^{\lambda X}]
    \le
    \exp\!\left(2\bigl(e^{\theta(\lambda)}-1\bigr)\right),
    \qquad
    \theta(\lambda):=
        pW(e^{\lambda}-1).
\end{equation}
\end{lemma}
\begin{proof}
Fix $\lambda\ge0$ and let $c=p(e^{\lambda}-1)$ and $\theta=cW$.
Consider a successful non-special round $r$. Conditioned on its
preceding history, let $\mathcal L_r$ be the set of eligible bundle indices and
$h_r:=|\mathcal L_r|\ge(n-t_r)/2$. For each $\ell\in \mathcal L_r$, let $u_\ell:=\sum_{e\in P_\ell^r\cap R_r}x_e$. The candidate bundles are disjoint, so $\sum_{\ell\in \mathcal L_r}u_\ell\le W$.

Let $\Delta_r$ be the weight allocated in this round. Conditional
on selecting bundle $\ell$, the sampling in Step~\ref{step:xos-sample} satisfies
upper-cylinder domination with parameter $p$. 
Lemma~\ref{lem:weighted-mgf} therefore gives
\[
    \mathbb{E}[e^{\lambda\Delta_r}
        \mid\text{history},\,\ell\text{ selected}]
    \le e^{cu_\ell}.
\]Averaging over the uniformly selected bundle, we obtain
\begin{align*}
    \mathbb{E}[e^{\lambda\Delta_r}\mid\text{history}]
    &\le 1+\frac1{h_r}\sum_{\ell\in 
    \mathcal L_r}(e^{cu_\ell}-1)\\
    &\le 1+\frac{e^\theta-1}{h_r}
    \le \exp\!\left(\frac{e^\theta-1}{h_r}\right).
\end{align*}
We will now apply Lemma~\ref{lem:conditional-mgf-composition} over all the rounds, taking both the LHS and RHS to be zero in special rounds, in a round reporting failure, and in all subsequent rounds. To bound the accumulated charge, we note that for any $r$, $n-t_r$ is at least as large as the total number of non-special rounds. Therefore,
\[
    \sum_{\substack{r:\,\text{successful}\\
                         \text{non-special round}}}\frac1{h_r}
    \le
    2\sum_{\substack{r:\,\text{successful}\\
                         \text{non-special round}}}\frac1{n-t_r}
    \le2.
\]
The accumulated charge is thus at most $2(e^\theta-1)$ on every realization. Applying Lemma~\ref{lem:conditional-mgf-composition} proves the lemma.
\end{proof}

\paragraph{A deterministic shadow allocation process.} We now handle the effect of the special rounds by defining a {\em deterministic} shadow process that depends on agents' preferences but runs independently of the algorithm's random choices. The shadow process proceeds as follows using the same global ordering of goods and agents as in Algorithm~\ref{alg:XOS}. Let $\shadows_i$ denote the goods allocated in the shadow process before round $i$. Initially $\shadows_1=\emptyset$. When agent $i$ arrives, we allocate the first good $e$ in the global ordering over $[m]\setminus \shadows_i$ with 
$v_i(\{e\})\ge\tau$, if such a
good exists; otherwise, skip the agent. $\shadows_{n+1}$ denotes the set of items allocated by this process.

At a high level, our analysis will limit the effect of special round allocations as follows. First, since shadow process allocations $\shadows_{n+1}$ are deterministically fixed, we can easily account for the reduction in the agent's value caused by allocating these items. Second, we argue that the algorithm allocates few items in special rounds outside of the set $\shadows_{n+1}$.

Recall that in the algorithm, the goods allocated in special and non-special rounds before round $i$ are denoted $\spc_i$ and $\nspc_i$ respectively. Observe that when an agent $i$ arrives in the algorithm, if the item allocated to it by the shadow process is still available, then the algorithm will allocate the same item to the agent as a special round item. Therefore, the number of special round items allocated outside of $\shadows_{n+1}$ by the algorithm is bounded by the number of non-special items in $\shadows_{n+1}$ allocated by the algorithm. We can then bound this number by noting that every item has low probability of being allocated in any individual non-special round.

To formalize this argument, let $d_i=|\shadows_i|$ denote the number of shadow items allocated before round $i$, $q_i:=n-d_i$, and $Q:=q_{n+1}$. Note that $q_i\ge n-i+1$ for all $i$.
Define the random quantities
\[
    Z_i:=|\shadows_i\cap \nspc_i|,
    \qquad
    Z:=|\shadows_{n+1}\cap \nspc|,
    \qquad
    a_i:=|\spc_i\setminus \shadows_i|.
\]
Thus $Z$ counts shadow allocations removed by actual non-special allocations, while $a_i$ counts special goods allocated by the algorithm outside the corresponding shadow prefix. Our goal is to show that $a_{n+1}$ (extra special allocations in the algorithm) is small relative to $Q$. We will show this by bounding $a_{n+1}$ by $Z$, and bounding $Z$ in terms of $Q$.

\begin{lemma}[Coupling the special allocations]
\label{lem:xos-virtual-coupling}
For every $i\in\{1,\ldots,n+1\}$ such that all rounds before $i$
complete without a reported failure,
\begin{equation*}
    \shadows_i\subseteq \spc_i\cup \nspc_i,
    \qquad
    t_i\le d_i,
    \qquad
    n-t_i=q_i+Z_i-a_i,
    \qquad
    0\le a_i\le Z_i\le Z.
\end{equation*}
Moreover,
\begin{equation*}
    \Pr[Z>Q/{50}]\le {101}p.
\end{equation*}
\end{lemma}

\begin{proof}
We first prove the containment $\shadows_i\subseteq \spc_i\cup \nspc_i$ by induction over $i$. At the beginning, all sets are empty. Suppose the claim holds before round $r$. Every good available in the algorithm is
then also available in the shadow process. If the shadow process allocates a good $g$, and $g$ is available in the algorithm, then 
the common global ordering forces the algorithm to allocate the same good. If $g$ is not available,
then $g\in \spc_r\cup \nspc_r$ already. Thus the containment is preserved
after each completed round. The same argument implies $t_i\le d_i$. Moreover, since $\spc_i$ and $\nspc_i$ are disjoint, we get
$d_i=|\shadows_i\cap \spc_i|+|\shadows_i\cap \nspc_i|=t_i-a_i+Z_i$. Rearranging proves the identity $n-t_i=q_i+Z_i-a_i$, and $t_i\le d_i$ implies $a_i\le Z_i$. Finally, $Z_i\le Z$ because $\shadows_i\subseteq \shadows_{n+1}$ and $\nspc_i\subseteq \nspc$.

Now we will bound $Z$ in terms of $Q$. First, if $Q=0$, the shadow process allocates a large item to all agents. Then, the algorithm mirrors the shadow process exactly and also allocates a large good to every agent. We have $Z=0$ and the claimed probability holds.

Henceforth assume $Q>0$. Let $n'$ be the number of rounds completed in the algorithm before failure. Note that with probability $1$, the number of non-special rounds in the algorithm before failure is at most $n'-|\spc_{n'+1}| \le n'-|\shadows_{n'+1}| +|\shadows_{n'+1}| -|\spc_{n'+1}|\le n'-|\shadows_{n'+1}| +|\shadows_{n'+1}\setminus \spc_{n'+1}| \le Q+Z$.

Consider a single non-special round $r$ in the algorithm and condition on the history before $r$. Every good in $\shadows_r$ has already been allocated by the algorithm. If the shadow process allocates a good in round $r$,
that good is unavailable in the algorithm, since the agent has
no available good of value at least $\tau$. Therefore every currently
available good in $\shadows_{n+1}$ lies in
$\shadows_{n+1}\setminus \shadows_{r+1}$, a set of size at most $n-r$. Each such good is allocated by the algorithm with probability at most $p/h_r$ where $h_r$ is the number of bundles of agent $r$ that are eligible in this round. Therefore, the expected increase in $Z$ is at most
\[
    \frac{p(n-r)}{h_r}
    \le\frac{2p(n-r)}{n-t_r}
    \le2p,
\] 
Summing over all non-special rounds, we therefore obtain 
\[
    \mathbb{E}[Z] \le2p\,\mathbb{E}[\text{number of non-special rounds}] \le2p\bigl(Q+\mathbb{E}[Z]\bigr).
\]
This yields
\[
    \mathbb{E}[Z]\le\frac{2p}{1-2p}\,Q.
\]
Markov's inequality then gives
\[
    \Pr[Z>Q/50]
    \le\frac{50\mathbb{E}[Z]}Q
    \le\frac{100p}{1-2p}
    \le\frac{101}{k}.
\]
where the last inequality follows because $k \geq 1500$.
\end{proof}

\paragraph{Reference bundles and the overall failure probability.} We will now focus on a single agent $i$. Note that the deterministic set $\shadows_i$ touches at most
$d_i$ of the agent's original MMS bundles. Fix a set of indices
$\mathcal{B}_i\subseteq[n]$ of exactly $q_i=n-d_i$ bundles that
are disjoint from $\shadows_i$, breaking ties deterministically.
These reference bundles are independent of the actual algorithm's random choices.

For each $h\in\mathcal{B}_i$, fix a nonnegative supporting additive
clause $c_{i,h}$ for $P_h^i$, and define
\[
    w_{i,e}:=
    \begin{cases}
        c_{i,h}(\{e\})/v_i(P_h^i),
            & e\in P_h^i,\ h\in\mathcal{B}_i,\\
        0, & \text{otherwise}.
    \end{cases}
\]
Define the total reference weight removed in non-special rounds by $X_i:=\sum_{e\in \nspc}w_{i,e}$.
The weights are fixed and have total weight $q_i\in[1,n]$.
We are now ready to prove the main theorem. In the following proof, we prove Theorem \ref{thm:xos-alg} for large $n$. The proof below captures the essential idea of establishing the guarantee. Optimizing the probability bounds gives us the theorem for all $n$. For cleaner exposition, we keep the proof for large $n$ here and give the full proof in Appendix \ref{app:xos-upper-bound}.
\begin{proof}[Proof of Theorem~\ref{thm:xos-alg} for large $n$]
Consider the event
\[
    \mathcal{E}:=
    \{Z\le Q/{50}\}
    \cap
    \bigcap_{i=1}^n\{X_i\le {6q_i/25}\}.
\]
First, we show that the algorithm cannot report failure on $\mathcal{E}$.
Fix any agent $i$ reached by the algorithm.
On $\mathcal{E}$, at most $2X_i\le 12q_i/25$ reference bundles
lose more than half their supporting weight to preceding non-special
rounds. Also, at most $a_i$ reference bundles are touched by 
special rounds: every such bundle avoids $\shadows_i$, so it must contain a good in $\spc_i\setminus \shadows_i$.

Therefore, the number of reference bundles that are (1) untouched by special rounds and (2) lose at most half their supporting weight to preceding non-special rounds is at least
\[
    q_i-{\frac{12q_i}{25}-a_i
    =\frac{13q_i}{25}-a_i}
\]
We now argue that this number is at least $(n-t_i)/2$, and consequently, the algorithm cannot report failure at $i$. Using the identity $n-t_i=q_i+Z_i-a_i$ from Lemma~\ref{lem:xos-virtual-coupling}, and noting that event  $\mathcal{E}$ implies {$a_i\le Z_i\le Z\le Q/50\le q_i/50$}, we get,
\[
    \left({\frac{13q_i}{25}}-a_i\right)-\left(\frac{n-t_i}{2}\right)
    = \frac{13q_i}{25} - a_i - \frac{q_i+Z_i-a_i}{2} = \frac{q_i}{50} - \frac{a_i+Z_i}{2}
    \ge0.
\]
We can therefore bound the failure probability of Algorithm \ref{alg:XOS} by taking a union bound over the complement of $\mathcal{E}$.

We want to bound the $\Pr[X_i > 6q_i/25]$ for all $i \in [n]$. Note that since $q_i$ bundles are independent of the randomness of the Algorithm and $1 \leq q_i \leq n$ for all $i$, we can apply Lemma \ref{lem:xos-weighted-loss} bounding the MGF of random loss $X_i$. First we give a simple proof that demonstrates how to use the lemma to get an asymptotic guarantee.

Choose $\lambda:=\min\{1,\frac{25\ln(6q_i/25+1)}{2q_i}\}$. By the convexity of the exponential function, we have $e^{\lambda}-1\le (e-1)\lambda$, and, 
\[\theta=pq_i(e^\lambda-1)\le (e-1)pq_i\lambda\le 13(e-1)p\ln(6q_i/25+1).\]

Noting $p=1/\lfloor x \rfloor$ and $q_i\le n$, we get $\theta=O(1)p\ln(6q_i/25+1)=O(1)$. So, using the convexity of the function $u\mapsto\exp(2(e^u-1))-1$, we have $\exp(2(e^\theta-1))-1\le C'\theta$ for a suitable absolute constant $C'$.

Finally, applying Markov's inequality to the nonnegative random variable
$e^{\lambda X}-1$ then gives
\begin{align*}
    \Pr[X > 6q_i/25]
    &\le
    \frac{\mathbb{E}[e^{\lambda X}-1]}{e^{\lambda 6q_i/25}-1} 
    \le \frac{\exp(2(e^\theta-1))-1}{e^{\lambda 6q_i/25}-1}
    \le\frac{C''p\ln(6q_i/25+1)}{(6q_i/25+1)^3}
    \le\frac{Cp}{(q_i+1)^2}.
\end{align*}
where $C''$ and $C$ are suitable constants.
The final inequality uses $\ln(x+1)\le x+1$.

Finally, taking a union bound over the complement of $\mathcal{E}$
and using Lemma~\ref{lem:xos-virtual-coupling},
\begin{align*}
    \Pr[\text{Algorithm~3 reports failure}]
    &\le\Pr[Z>Q/50]
        +\sum_{i=1}^n\Pr[X_i>6q_i/25]\\
    &\le101p+Cp\sum_{i=1}^n\frac1{(q_i+1)^2}\\
    &\le101p+Cp\sum_{i=1}^n\frac1{(n-i+2)^2}
    \le(101+C)p
    =O\!\left(\tau\right).
\end{align*}
Together with Lemma~\ref{lem:xos-allocation-value}, this proves
the theorem for large $n$. 
\end{proof}
\section{Efficient implementation of our algorithms}
\label{sec:comp}
Algorithms \ref{alg:submod-rr} and \ref{alg:XOS} assume normalized input. Thus, each arriving agent’s MMS value must be known to normalize their valuation (Appendix \ref{app:norm}). Given these values and value-oracle access, Algorithm \ref{alg:submod-rr} can be implemented with polynomially many oracle calls, including those used to compute marginal values.
Algorithm \ref{alg:XOS} additionally requires an MMS partition for each agent
and XOS-oracle access to obtain supporting additive clauses; sorting
the resulting item weights and grouping the goods takes polynomial
time. Thus, the remaining computational bottleneck is obtaining
the MMS values and, for Algorithm \ref{alg:XOS}, the corresponding partitions.
Exact computation can be replaced by approximation: given estimates
$\widehat{\mu}_i$ satisfying
$\beta\,\mathrm{MMS}_i \leq \widehat{\mu}_i \leq \mathrm{MMS}_i$
and, for Algorithm \ref{alg:XOS}, partitions whose every bundle has value at
least $\widehat{\mu}_i$, an $\alpha$-MMS guarantee becomes an
$\alpha\beta$-MMS guarantee after scaling the thresholds accordingly.
For additive valuations, a $(1-\varepsilon)$-approximation to the
MMS value can be computed in polynomial time for every fixed
$\varepsilon>0$ using the PTAS of \citet{woeginger1997ptas}.
For submodular valuations, a $1/3$-approximation can be computed
in polynomial time using value queries
\citep[Corollary~5.5]{ghodsi2022beyond}.
For XOS valuations, a $1/8$-approximate MMS partition, and hence
a $1/8$-approximation to the MMS value, can be computed in polynomial
time using demand and XOS queries
\citep[Corollary~6.3]{ghodsi2022beyond}.
Consequently, both algorithms admit polynomial-time implementations
with polynomially many calls to the respective oracles, with only
constant-factor losses in their approximation guarantees.
\section{Related Work}\label{sec:related_work}
Fair division, the problem of allocating scarce resources ``fairly'' among heterogeneous agents, has long been a central topic across computer science, economics, social choice theory, and operations research. Over the past several decades, it has been extensively studied in offline settings, where all agents and their valuations are known in advance (e.g., see latest surveys \cite{moulin2019fair,amanatidis2023fair} and books \cite{FD-book1,FD-book2}). More recently, motivated by practical applications such as job and task scheduling or disaster relief and food banks, attention has shifted to online settings in which goods or agents arrive over time. In this section, we highlight previous work that is most closely related to our setting. We refer the reader to excellent surveys \cite{aleksandrov2020online, amanatidis2023fair} for further references.

\medskip

\noindent \textit{Online fair division with agent arrivals.} To the best of our knowledge, \citet{kulkarni2025online} is the only paper that considers the agent arrival setting with indivisible goods prior to our work. As discussed in Section \ref{sec:intro}, the paper considers the stochastic case where the valuation function of the arriving agent is drawn from a finite distribution over known additive $k$ agent types. No results are known when the allocator has no knowledge of future agents' valuation functions that are chosen by an (oblivious) adversary.

Other works on online fair division with agent arrivals focus on the setting where goods are divisible. \citet{walsh2011online} considers the setting of dividing a single divisible good to $n$ agents arriving online. \citet{kash2014no} studied a setting called dynamic fair division to fairly allocate a set of divisible goods to agents who arrive but never depart from the system. Works by \citet{sinclair2022sequential} and \citet{banerjee2023online} studied the setting of allocating a set of goods to agents which are classified by their valuation functions.
\citet{friedman2015dynamic, friedman2017controlled} consider algorithms that simultaneously achieve fairness notions and minimize disruptions where disruption is defined when a reallocation of a previously allocated resource happens.
In a set of related works, 
\citet{bogomolnaia2019simple} consider the setting where agents' valuation functions are drawn from a distribution, while \citet{donahue2020fairness} study the setting where an uncertain number of agents from different "groups" arrive. They both define the notion of fairness as agents from different groups receive resources with equal probability.

\medskip

\noindent \textit{Online fair division with indivisible items.} While most of the current literature in online fair division focuses on divisible items, there are a few papers dealing with indivisible items. These mostly focus on the item-arrival setting. \citet{benade2018make} study the fairness notion of envy, minimizing the total envy over a certain period of time. \citet{he2019achieving} minimize the total number of re-allocations needed to achieve EF1. \citet{zhou2023multi} study MMS in the item-arrival setting for indivisible goods and chores.
\citet{zeng2020fairness} focus on the setting where agents' values are drawn from a random distribution, and study the fairness and efficiency trade-offs over time. \citet{procaccia2024honor} and \citet{yamada2024learning}
study how to allocate indivisible online items through the lens of bandit learning. A recent work \citet{halpern2025online} connects envy minimization in online indivisible item arrival setting to discrepancy theory.

\medskip

\noindent \textit{Subsequent Work.} Following our arXiv upload \cite{chawla2026singleorsampleonlinefairallocation} on September 26, 2026, \citet{li2026randomizedalgorithmmmsallocation} posted a manuscript dated October 11, 2026, studying online fair division for {\em chores} under {\em additive} valuations and agent arrivals, and obtained a constant MMS guarantee to all agents simultaneously with constant probability. They also extend their algorithm to goods with additive valuations, obtaining a $c$-MMS guarantee with probability $1-85c^2$ for $c\le 1/15$. Their analysis also uses an MGF-based argument, a technique shared by the two works. Their guarantee is stronger than ours for additive valuations. Our work focuses on submodular and XOS valuations, and we have not sought to optimize the constants, keeping the exposition centered on the main ideas.


\section{AI Disclosure} The algorithms and results in this paper were developed by the authors over the course of more than a year. The main qualitative results were obtained without using AI. Some of the analysis ideas, proofs, and constants were later improved with the help of GPT 5.6 Sol and GPT 6 Astra. The XOS lower bound was obtained with significant help of GPT 5.6 Sol. All proofs have been verified and rewritten by the authors to improve clarity and convey intuition. The authors take full responsibility for the correctness of all claims.


\newpage

\appendix
\section*{Appendix}
\addcontentsline{toc}{section}{\\ Appendix}
\section{Normalization of Valuation Functions}

\label{app:norm}

Given any valuation function $v_i$, consider a 
new valuation function $v'_i$ defined as $v'_i(S) = \frac{1}{\text{MMS}^n_{v_i}} v_i(S)$ for all $S\subseteq [m]$.\footnote{If $\text{MMS}^n_{v_i}=0$ then agent $i$ need not be given any item and hence can be simply ignored.} 
For the original MMS partition, this scales each bundle to have value at least $1$ and the smallest of these bundles has value $1$. On the other hand, if the normalized instance has a new MMS partition with the value of each bundle being more than $1$, the original instance has a partition where the value of each bundle is more than $\MMS_i^n$, which is not possible. Therefore, it follows that $\MMS^n_{v'_i}=1$ and for any bundle $B$, $v'_i(B)\ge \alpha$ implies $v_i(B) \ge \alpha \MMS^n_{v_i}$. Since $v_i'$ scales the value of every set by the same positive constant
$1/\mathrm{MMS}_i$, it preserves the agent's ranking of sets and of marginal
values. Hence the agent's safe strategy is unaffected by this normalization.


\section{Exponential Moment Bounds for Iterative Adaptive Sampling}
\label{sec:sampling-ext}

In this section we prove the exponential moment bounds of Lemmas~\ref{lem:iterative-cylinder} and \ref{lem:conditional-mgf-composition}. As discussed previously, the adaptivity of the iterative sampling process can destroy negative dependence of the selection across different elements. Nevertheless, we show that an upper-cylinder bound survives both adaptive bucketing and repeated sampling. This can,
in turn, be used to bound the moment generating function of the desired random variable, from which concentration follows using standard techniques.

\subsection{Upper cylinder domination}

We first show that for the iterative sampling without replacement process, the probability of selecting every element of any fixed set is no greater than the corresponding probability under an independent Bernoulli process in which each element is selected independently with probability \(p\). We refer to this property as {\em upper-cylinder domination by an independent Bernoulli process}.



\begin{lemma}
\label{lem:bucket-cylinder}
    Let $G$ be nonempty and $A\sim\operatorname{SwR}(G,p)$. For every fixed $T\subseteq G$,
\[
\Pr[T\subseteq A]\le p^{|T|}.
\]
\end{lemma}

\begin{proof}
Let $m=|G|$ and suppose first that $pm$ is an integer. Then $A$ is a uniformly random subset of $G$ of size $y=pm$. Then, for a fixed set $T$, the probability that all items in $T$ are sampled into $A$ is
\[
\Pr[T\subseteq A]=\frac{\binom{m-|T|}{y-|T|}}{\binom{m}{y}}
=\frac{(y)_{|T|}}{(m)_{|T|}}\le \left(\frac{y}{m}\right)^{|T|},
\]
where $(r)_{|T|}:=r(r-1)\cdots(r-|T|+1)$ is the falling factorial.

For the general case, recall that we select a random subset of size either $y$ or $y+1$ where $y = \lfloor pm\rfloor$. We choose size $y+1$ with probability $x=pm -\lfloor pm\rfloor$ and $y$ with probability $1-x$.
Consider a fixed subset $T$ with $|T|=q$. The probability that all items in $T$ are sampled into $A$ is
\[
\Pr[T\subseteq A]=(1-x)\frac{\binom{m-q}{y-q}}{\binom{m}{y}}+x\frac{\binom{m-q}{y+1-q}}{\binom{m}{y+1}}
=(1-x)\frac{(y)_q}{(m)_q}+x\frac{(y+1)_q}{(m)_q}.
\]
It then remains to show that
\[
(1-x)\frac{(y)_q}{(m)_q}+x\frac{(y+1)_q}{(m)_q}\le \left(\frac{y+x}{m}\right)^q=p^q.
\]
We will now prove the inequality.
If $q=y+1$, then $(y)_q=0$ and $(y+1)_q=q!$. Since
\[
\frac{(m)_q}{m^q}
=\prod_{j=0}^{q-1}\left(1-\frac{j}{m}\right)
\ge \frac{q!}{q^q},
\]
and AM--GM gives $q x^{1/q}\le q-1+x$, we obtain
\[
x\frac{q!}{(m)_q}
\le \frac{xq^q}{m^q}
\le \left(\frac{q-1+x}{m}\right)^q
=\left(\frac{y+x}{m}\right)^q.
\]
Henceforth assume $q\le y$.

Let $Y = (y)_q/(m)_q$ and $Y^+ = (y+1)_q/(m)_q$. Consider the function 
\[
f(x) := (1-x)Y + xY^+ - (y+x)^q/m^q.
\]
Our goal is to show that $f(x)\le 0$ for $x\in [0,1]$. 
If $q=1$, then $f\equiv 0$. Hence assume $q\ge 2$.
The function $f$ is concave on $[-y,\infty)$; let $x^*>-y$
be its stationary point, we have $f(x) \leq f(x^*)$.

By writing the first order equation, we get 
\[\frac{(x^*+y)^{q-1}}{m^q} = \frac{Y^+-Y}{q}\]
Moreover, since $y+1\le m$,
\[
(x^*+y)^{q-1}
= m^{q-1}\prod_{j=1}^{q-1}\frac{y+1-j}{m-j}
\le (y+1)^{q-1}.
\]
Thus $x^*\le 1$.
Furthermore, 
\[Y^+-Y = (y+1)\frac{(y)_{q-1}}{(m)_q} - (y-q+1)\frac{(y)_{q-1}}{(m)_q} = \frac{q}{y-q+1} Y\]

Substituting these back in $f$ and simplifying, we get,
\begin{align*}
f(x^*) & = Y + (Y^+-Y)x^* - (Y^+-Y)(x^*+y)/q\\
    & = Y\left( 1+\frac{q}{y-q+1}\left(x^* - \frac{x^*}{q}- \frac{y}{q}\right)  \right)\\
    & = \frac{Y}{y-q+1} (q-1)(x^*-1) \le 0
\end{align*}
where the final inequality uses the fact that $1\le q< y+1$ and $x^*\le 1$.
\end{proof}

We next extend upper cylinder domination to a single iteration of sampling without replacement with multiple adaptively defined buckets.
\begin{lemma}[Adaptive bucketing]
\label{lem:adaptive-cylinder}
Let $G$ be a universe and $A$ be drawn from $G$ through a multi-bucket SwR process with probability bound $p$. 
Then, for every fixed $T\subseteq G$,
\[
\Pr[T\subseteq A]\le p^{|T|}.
\]
\end{lemma}

\begin{proof}
We induct on $|G|$, uniformly over the history preceding the iteration.
Condition on the first bucket $B^{(1)}$ and its sampling probability $p^{(1)}\le p$,
and let $A^{(1)}$ be its sample. For every realization of $A^{(1)}$, the remaining
process acts only on $G\setminus B^{(1)}$. By induction, its probability of
sampling all of $T\setminus B^{(1)}$ is at most $p^{|T\setminus B^{(1)}|}$, regardless
of its dependence on $A^{(1)}$. Hence, by Lemma~\ref{lem:bucket-cylinder},
\[
\Pr[T\subseteq A\mid B^{(1)}, p^{(1)}]
\le \Pr[T\cap B^{(1)}\subseteq A^{(1)}\mid B^{(1)}, p^{(1)}]\cdot p^{|T\setminus B^{(1)}|}
\le p^{|T\cap B^{(1)}|}p^{|T\setminus B^{(1)}|}=p^{|T|}.
\]
Averaging over $B^{(1)}$ and $p^{(1)}$ completes the induction.
\end{proof}
Finally, we show that the upper cylinder domination composes across iterations.

\begin{lemma}[Composition across iterations]
\label{lem:iterative-cylinder-comp}
Consider an iterative SwR process with parameters $p_1,\ldots,p_n$ with ground set $G$ and sampled set $A$.
Let $q=\prod_{i=1}^n(1-p_i)$.
Then, for every fixed $T\subseteq G$,
\[
\Pr[T\subseteq A]\le(1-q)^{|T|}.
\]
\end{lemma}

\begin{proof}
Let $\A_i=G\setminus G_{i+1}$ and
$r_i=1-\prod_{j=1}^i(1-p_j)$, with $\A_0=\varnothing$ and $r_0=0$.
We inductively prove $\Pr[T\subseteq \A_i]\le r_i^{|T|}$ for every $T$.
Condition on the full history $\mathcal F_{i-1}$ before iteration $i$.
By Lemma~\ref{lem:adaptive-cylinder},
\[
\Pr[T\subseteq \A_i\mid\mathcal F_{i-1}]
\le p_i^{|T\setminus \A_{i-1}|}
=\prod_{g\in T}\bigl(p_i+(1-p_i)\mathbf 1_{\{g\in \A_{i-1}\}}\bigr).
\]
Expanding this product and applying the induction hypothesis gives
\begin{align*}
\Pr[T\subseteq \A_i]
&\le\sum_{U\subseteq T}p_i^{|T|-|U|}(1-p_i)^{|U|}
                         \Pr[U\subseteq \A_{i-1}]\\
&\le\bigl(p_i+(1-p_i)r_{i-1}\bigr)^{|T|}
 =r_i^{|T|}.
\end{align*}
Taking $i=n$ proves the upper-cylinder bound. 
\end{proof}

\subsection{From upper cylinder domination to MGF bounds}
We now use upper-cylinder domination to obtain an MGF bound for weighted sums.

\begin{lemma}[Weighted MGF bound under upper-cylinder domination]
\label{lem:weighted-mgf}
Let $G$ be a finite ground set and $A \subseteq G$ a random subset.
Suppose that, for some $p \in [0,1]$,
$\Pr[T \subseteq A] \le p^{|T|}$ for every fixed $T \subseteq G$. Let $(x_e)_{e\in G}$ be fixed weights with $x_e \in [0,1]$. 
Define $X := \sum_{e\in A} x_e$ and $W := \sum_{e\in G} x_e$. Then, for every $\lambda \ge 0$,
\begin{equation}
\label{eq:weighted-mgf-cylinder}
    \mathbb{E}[e^{\lambda X}]
    \le
    \prod_{e\in G}
        \left(1+p(e^{\lambda x_e}-1)\right)
    \le
    \exp\left(pW(e^{\lambda}-1)
    \right).
\end{equation}
\end{lemma}

\begin{proof}
For each $e\in G$, let $X_e=\mathbbm{1}_{\{e\in A\}}$ and set
$z_e=e^{\lambda x_e}-1\ge0$. Since $X_e\in\{0,1\}$,
\[
    e^{\lambda X}
    = \prod_{e\in G} e^{\lambda x_e X_e} =
    \prod_{e\in G}(1+z_e X_e)
    =
    \sum_{T\subseteq G}
        \left(\prod_{e\in T}z_e\right)
        \mathbbm{1}_{\{T\subseteq A\}}.
\]
Taking expectations and applying the upper-cylinder bound gives
\[
    \mathbb{E}[e^{\lambda X}]
    \le
    \sum_{T\subseteq G}
        \left(\prod_{e\in T}z_e\right)p^{|T|}
    =
    \prod_{e\in G}(1+pz_e).
\]
This proves the first inequality in
\eqref{eq:weighted-mgf-cylinder}.
For the second, convexity implies that $e^{\lambda x_e}-1\le x_e(e^{\lambda}-1)$ for $x_e\in [0,1]$.
Consequently, using $1+u\le e^u$,
\[
    \prod_{e\in G}(1+pz_e)
    \le
    \exp\left(p\sum_{e\in G}z_e\right)
    \le
    \exp\left(pW(e^{\lambda}-1)
    \right).
\]
\end{proof}

Lemmas~\ref{lem:iterative-cylinder-comp} and \ref{lem:weighted-mgf} together imply Lemma~\ref{lem:iterative-cylinder}. We now prove Lemma~\ref{lem:conditional-mgf-composition}, which is restated for convenience.

\condmgf*

\begin{proof}
Define
\[
    M_j:=
    \exp\left(
        \lambda Y_j-\sum_{r=1}^j b_r
    \right),
    \qquad j=0,\ldots,N.
\]
Since $b_j$ is determined by $\mathcal{F}_{j-1}$,
\begin{align*}
    \mathbb{E}[M_j\mid\mathcal{F}_{j-1}]
    &=
    M_{j-1}e^{-b_j}
    \mathbb{E}\left[
        e^{\lambda(Y_j-Y_{j-1})}
        \,\middle|\,\mathcal{F}_{j-1}
    \right]\\
    &\le M_{j-1}.
\end{align*}
Thus $(M_j)$ is a nonnegative supermartingale with $M_0=1$.
Taking expectations and iterating proves
\eqref{eq:composed-conditional-mgf}.
To obtain the stopped version, replace the increment
$Y_r-Y_{r-1}$ and the charge $b_r$ by
\[
    \mathbf{1}_{\{\tau\ge r\}}(Y_r-Y_{r-1})
    \qquad\text{and}\qquad
    \mathbf{1}_{\{\tau\ge r\}}b_r,
\]
respectively.
Because $\tau$ is a stopping time,
$\mathbf{1}_{\{\tau\ge r\}}$ is
$\mathcal{F}_{r-1}$-measurable.
The modified increment and charge therefore still satisfy
\eqref{eq:one-round-conditional-mgf}: on $\{\tau<r\}$ both
are zero, and on $\{\tau\ge r\}$ the original bound applies.
Applying \eqref{eq:composed-conditional-mgf} to this modified
process at time $N$ gives
\[
    \mathbb{E}\left[
        \exp\left(
            \lambda Y_\tau-\sum_{r=1}^{\tau}b_r
        \right)
    \right]\le1.
\]

Finally, if the accumulated charge is at most $B$ almost surely, then
\[
    e^{\lambda Y_\tau}
    =
    \exp\left(
        \lambda Y_\tau-\sum_{r=1}^{\tau}b_r
    \right)
    \exp\left(\sum_{r=1}^{\tau}b_r\right)
    \le
    e^B
    \exp\left(
        \lambda Y_\tau-\sum_{r=1}^{\tau}b_r
    \right).
\]
Taking expectations proves the claim.
\end{proof}

\section{XOS Lower Bound}
\label{app:xos-lb}

We restate our lower bound here.

\xoslb*


By Yao's minimax principle~\cite{yao1977probabilistic}, it suffices
to construct one input distribution under which every deterministic online
algorithm succeeds with probability $o(1)$, with a bound independent of the
algorithm. 
Before describing our general construction, we begin with a simple construction that rules out $2/3$-\MMS\ over binary XOS valuations.

\subsection{Warm-up: a two-round construction.}
 There are $n$ agents and $m$ goods.
Assume for simplicity that all quantities below are integral. All random
choices defining the instance are made in advance, independently of
the algorithm, so the adversary is oblivious. The agents arrive in two rounds of $n/2$ agents each.
\begin{itemize}[leftmargin=*]
    \item \textbf{First round.}
    Partition the $m$ goods uniformly at random into $n$ equal-sized
    bins $B_1,\ldots,B_n$. All first-round agents have valuation
    \(
        v(S)=\max_{j\in[n]} |S\cap B_j|
    \)
    and hence \MMS\ value $m/n$.

    \item \textbf{Second round.}
    Fix $k\le n/2$, to be chosen below. Select a uniformly random
    $k$-subset $Q=\{i_1,\ldots,i_k\}\subseteq[n]$; the corresponding
    bins form the \emph{pool}. Divide the second-round agents into
    $k$ groups of $n/(2k)$ agents. Independently partition each
    selected bin $B_{i_r}$ uniformly at random into $n$ equal-sized
    parts $B_{i_r}^1,\ldots,B_{i_r}^n$. Every agent in group $r$
    has valuation
    \(
        v_r(S)=\max_{j\in[n]} |S\cap B_{i_r}^j|
    \)
    and hence \MMS\ value $m/n^2$.
\end{itemize}
The pool and its internal partitions remain hidden until the second round.

If every first-round agent receives a $2/3$-\MMS\ allocation, each
must receive at least $2m/(3n)$ goods from a single bin. These bins must be distinct, since no bin can support two such
allocations. Thus at least $n/2$ bins each lose at least $2m/(3n)$
goods. Call these bins \emph{damaged}, and let $D\subseteq[n]$
index them. 

\paragraph{A damaged bin in the pool forces failure.}
Suppose the algorithm succeeds in the first round. Each damaged bin contains at most $m/(3n)$ unallocated goods. The internal partition of the $m/n$ goods is determined by the adversary independently of which $m/3n$ goods the algorithm leaves untouched, for an appropriate choice of $m$ and $n$, we can show that with probability $1-o(1)$, every part of every damaged bin in the pool contains fewer than $2m/(3n^2)$ unallocated goods. 
Consequently, no
second-round agent associated with such a bin can receive a
$2/3$-MMS\ allocation.

\paragraph{The pool contains a damaged bin with high probability.}
If every first-round agent succeeds, at least \(n/2\) bins are damaged. The set \(D\) is determined by the first-round partition and the algorithm’s first-round allocations. Since the pool and its internal partitions remain hidden during the first round, \(Q\) remains a uniformly random \(k\)-subset independent of \(D\). Therefore, the probability that all first-round agents succeed but the pool contains no damaged bin is at most \(2^{-k}\). Taking \(k=\Theta(\log n)\), and combining this with the concentration bound, gives an overall success probability of \(o(1)\).

\subsection{The complete construction}
We now extend these ideas to a multi-round allocation process to rule out an $\alpha$ approximation to \MMS\ for $\alpha=2/\log\log n$, as $n$ tends to infinity.
 See
\Cref{fig:pool-lottery-instance} for an illustration of the construction.

\paragraph{Binary-XOS agents and good allocations.} Each agent's valuation in our construction is defined by $n$ pairwise disjoint, equal-sized sets of items, which we call \emph{bins}: $B_1,\ldots,B_n$. Her value for a set $S$ is $v(S)=\max_{j\in[n]} |S\cap B_j|$. Her MMS value is therefore the size of a bin, and she receives an $\alpha$-MMS allocation precisely when she receives at least an $\alpha$ fraction of the items in some bin. 

\begin{figure}[ht!]
\centering
\begin{tikzpicture}[
  x=1cm, y=1cm, >=Latex,
  pool/.style={draw, thick},
  witness/.style={->, very thick, red!70!black},
  otherpath/.style={->, thick},
  roundlabel/.style={font=\large, anchor=east},
  smalllabel/.style={font=\small, align=center},
  process/.style={smalllabel, text width=3.2cm},
  pooltitle/.style={smalllabel, anchor=south west, inner sep=0pt},
  dots/.style={inner sep=0pt},
]

\node[roundlabel] at (1.65,10.60) {Round $0$};
\node[smalllabel] at (6.20,11.60) {one pool $P_0=G$};
\begin{scope}[shift={(2.2,10)}]
  \fill[pattern=north east lines, pattern color=red!75]
    (1,0) rectangle (2,0.55);
  \fill[pattern=north east lines, pattern color=red!50]
    (6,0) rectangle (7,0.32);
  \draw[pool] (0,0) rectangle (8,1.2);
  \foreach \x in {1,...,7}
    \draw (\x,0) -- (\x,1.2);
  \node[smalllabel] at (0.5,0.92) {$B_1$};
  \node[smalllabel] at (1.5,0.92) {$B_2$};
  \node[dots] at (4.5,0.65) {$\cdots$};
  \node[smalllabel] at (7.5,0.92) {$B_n$};
\end{scope}

\node[smalllabel,anchor=east,
      text width=5cm] at (9.2,9.15)
  {Global lottery:\\uniformly choose $q_1$ bins};
\node[process] at (6.65,7.80)
  {Freshly partition each retained bin\\into $n$ bins};

\draw[witness]
  (3.70,10.00) .. controls (3.70,8.95) and (4.90,8.05) .. (4.90,6.99);
\node[smalllabel, red!70!black, anchor=east,
      text width=2.5cm] at (3.10,8.60)
  {a damaged bin\\is retained};

\draw[otherpath]
  (8.70,10.00) .. controls (8.70,8.95) and (10.15,8.05) .. (10.15,6.99);
\node[smalllabel, anchor=west, text width=2.3cm] at (10.35,8.60)
  {another retained\\bin};

\node[roundlabel] at (1.65,6.48) {Round $1$};
\node[smalllabel, anchor=east] at (1.65,6.03) {$q_1$ pools};
\node[pooltitle] at (2.20,7.20) {witness pool $P_1$};
\node[pooltitle] at (8.00,7.20) {another pool};

\begin{scope}[shift={(2.2,5.9)}]
  \fill[pattern=north east lines, pattern color=red!35]
    (0,0) rectangle (3.6,0.28);
  \fill[pattern=north east lines, pattern color=red!75]
    (1.8,0.28) rectangle (2.4,0.77);
  \draw[pool, red!70!black] (0,0) rectangle (3.6,1.05);
  \foreach \x in {0.6,1.2,1.8,2.4,3.0}
    \draw (\x,0) -- (\x,1.05);
  \node[dots] at (0.9,0.65) {$\cdots$};
\end{scope}

\begin{scope}[shift={(8.0,5.9)}]
  \fill[pattern=north east lines, pattern color=red!30]
    (0,0) rectangle (2.6,0.18);
  \draw[pool] (0,0) rectangle (2.6,1.05);
  \foreach \x in {0.65,1.3,1.95}
    \draw (\x,0) -- (\x,1.05);
  \node[dots] at (0.975,0.60) {$\cdots$};
\end{scope}

\draw[witness] (4.30,5.90) -- (4.30,4.50);
\node[red!70!black] at (4.30,4.08) {$\vdots$};
\draw[witness]
  (4.30,3.65) .. controls (4.30,3.15) and (5.20,3.15) .. (5.20,2.34);

\node[smalllabel,anchor=east,
      text width=5cm] at (9.4,4.10)
  {Repeat the lottery\\and fresh random partition};

\draw[otherpath] (9.30,5.90) -- (9.30,4.50);
\node at (9.30,4.08) {$\vdots$};
\draw[otherpath]
  (9.30,3.65) .. controls (9.30,3.15) and (10.15,3.15) .. (10.15,2.34);

\node[roundlabel] at (1.65,1.83) {Round $d$};
\node[smalllabel, anchor=east] at (1.65,1.38) {$q_d$ pools};
\node[pooltitle] at (2.20,2.55) {witness pool $P_d$};
\node[pooltitle] at (8.00,2.55) {another pool};

\begin{scope}[shift={(2.2,1.25)}]
  \fill[pattern=north east lines, pattern color=red!70]
    (0,0) rectangle (3.6,0.78);
  \draw[pool, red!70!black] (0,0) rectangle (3.6,1.05);
  \foreach \x in {0.6,1.2,1.8,2.4,3.0}
    \draw (\x,0) -- (\x,1.05);
  \node[dots] at (0.9,0.92) {$\cdots$};
\end{scope}

\begin{scope}[shift={(8.0,1.25)}]
  \draw[pool] (0,0) rectangle (2.6,1.05);
  \foreach \x in {0.65,1.3,1.95}
    \draw (\x,0) -- (\x,1.05);
  \node[dots] at (0.975,0.60) {$\cdots$};
\end{scope}

\node[smalllabel, text width=6cm] at (4.00,0.5)
  {Every bin of the witness pool\\has $<\alpha n^2$ goods left.};

\draw[pattern=north east lines, pattern color=red!70]
  (2.20,-0.40) rectangle (2.60,-0.10);
\node[smalllabel, anchor=west] at (2.80,-0.25)
  {already allocated goods};
\draw[witness] (8.00,-0.25) -- (8.75,-0.25);
\node[smalllabel, anchor=west] at (8.95,-0.25)
  {witness sequence};

\end{tikzpicture}
\caption{
An illustration of the random-pool lower-bound instance.
In each round, every pool is partitioned into $n$ bins.
A global random selection process retains some bins, and each retained bin becomes
a pool in the next round and is freshly repartitioned.
Red hatching denotes goods already allocated.
The red arrows follow the single witness sequence used in the analysis:
successful service creates new damage, while each fresh random partition
spreads the accumulated damage among the next level's bins.
By the final round, every bin of the witness pool has fewer than
$\alpha n^2$ available items.
}
\label{fig:pool-lottery-instance}
\end{figure}
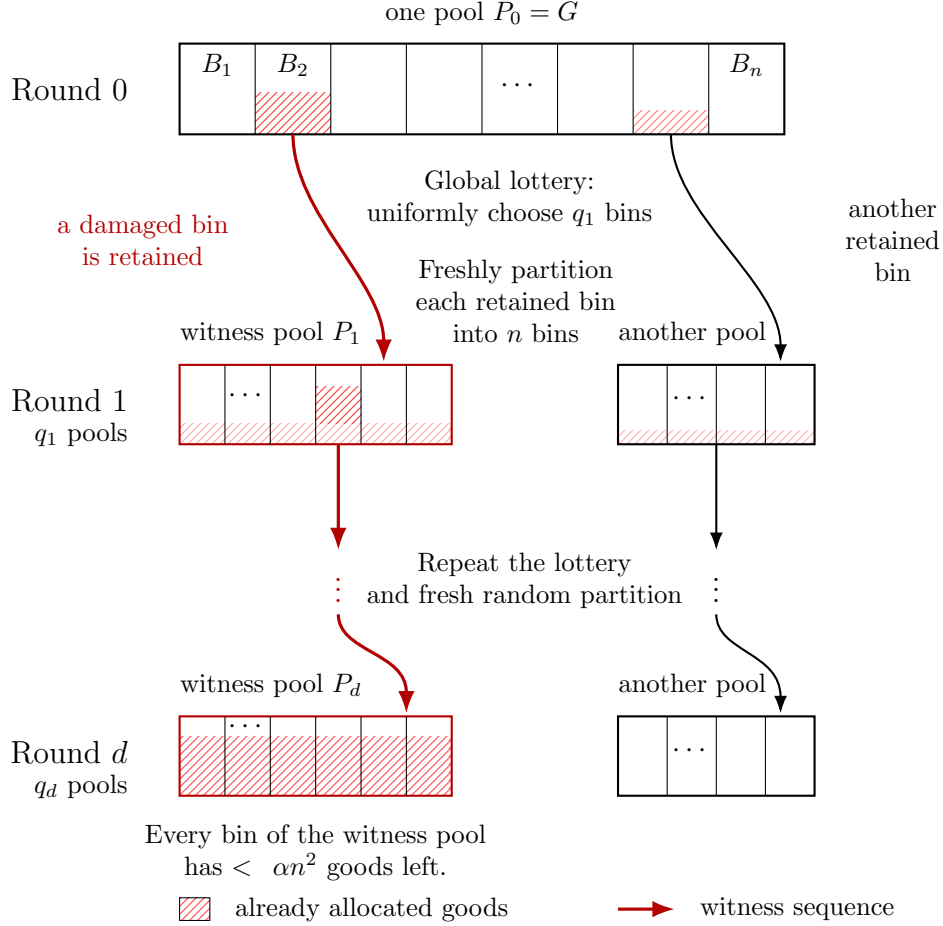

\paragraph{Rounds and pool counts.}
Let $d=\lceil1/\alpha\rceil$. There are $d+1$ rounds, indexed by
$r=0,\ldots,d$. Let $n$ be a sufficiently large integer of the form $(d+1)K^{2^{d}-1}$  where $K\in \mathbb{N}$ and denote $T=\frac{n}{d+1}$. Round $r$ has $q_r$ pools where
\begin{equation*}
\label{eq:xos-pool-counts}
    q_0=1,
    \qquad q_{r+1}=Kq_r^2 \; \text{ for } \; 0\le r<d, \quad \text{and} \quad K=T^{1/(2^d-1)}.
\end{equation*}
Thus $q_r=K^{2^r-1}$ and $q_r\le q_d=T$. We take $n$ sufficiently large so that, with $\alpha=2/\log\log n$, we have $K\ge 2\alpha^2 \log{n}$.

Each round has $T$ agents. The $T$ round-$r$ agents are divided evenly among the $q_r$ pools. Let $s_{r} := T/q_r=K^{2^d-2^r}$ denote
the number of agents assigned to a single pool $P$ in this round. 


\paragraph{Items, partitions, and valuations.} There are $n^{d+3}$ items. Each pool $P$ in a round $r$ is associated with a set of $nM_r$ items where $M_r=n^{d-r+2}$ for $0\le r\le d$. In particular, there is a single pool in round $0$ containing all $G=[nM_0]$ items. 

Each round-$r$ pool $P$ is partitioned uniformly at random into $n$ labeled bins
$B_{P,1},\ldots,B_{P,n}$ of size $M_r$ each. All $s_{r}$ agents assigned to this pool have valuation $v_P(S)=\max_{j\in[n]}|S\cap B_{P,j}|$. 

\paragraph{The random instance.}
We now describe how the adversary randomizes the instance. For each $r<d$, the adversary considers all of the $nq_r$ round-$r$ bins corresponding to the round-$r$ pools and chooses a uniformly random subset $Q_r$ of $q_{r+1}$ bins. These bins become
the round-$(r+1)$ pools and are randomly subdivided into $n$ bins each for round $r+1$. Observe that the size of the round-$(r+1)$ bins, $M_{r+1}$, is accordingly a $1/n$ fraction of the size $M_r$ of the round-$r$ bins.


\paragraph{Arrival order and information.}
Agents arrive in increasing round order, with consecutive pool groups in the
support order specified above. All partitions and pools are sampled
beforehand, defining an input distribution $\mathcal D_{n,\alpha}$ independent
of the algorithm. The algorithm knows $n$ and $G$, but learns valuations (and, accordingly, the selection and partition of items into pools and bins) only upon the arrival of the corresponding agents. 

\subsubsection{Analysis of deterministic algorithms}
\label{sec:xos-analysis-lb}

We will now outline the proof of Theorem~\ref{thm:xos-lower-bound} through a sequence of lemmas.

We first observe that in any instance generated as above, it is possible to guarantee each agent their \MMS\ value. An offline allocator can in every round allocate items to agents in that round from bins that are not chosen as pools for the following round. As there are $n$ bins in every pool in round $r$, $T/q_r\le \alpha n$ agents in the pool, and at most $q_{r+1}\ll n$ bins move to the next round, it is always possible to satisfy every agent through such an allocation. The online algorithm, however, must allocate before learning which bins are retained for the following rounds.


\begin{lemma}
\label{lem:xos-instances}
For every instance in the support of $\mathcal D_{n,\alpha}$, there exists an allocation $A_1, \cdots, A_n$ that guarantees to each agent $i$ a value of $v_i(A_i)\ge\MMS_i$.
\end{lemma}

\begin{proof}[Proof of Lemma~\ref{lem:xos-instances}]
The number of agents and number of items are easy to compute that are $n$ and $n^{d+3}$. According to the definition of agents' valuation function one can also verify 
\[
    \min_i\operatorname{MMS}_i\ge n^2,
    \qquad
    \max_{i,g}\frac{v_i(\{g\})}{\operatorname{MMS}_i}\le n^{-2}.
\]

To prove offline feasibility, note we are dividing the number of agents in each round as evenly as possible, we know the number of agents of each group $S_{r}$ satisfies
\begin{equation}
\label{eq:xos-rounded-groups}
\begin{split}
    s_{r}
    &
     \ge \frac{T}{q_r}
      = \frac{n}{(d+1)q_r}
     \ge \frac{n}{2(d+1)q_r},\\
    s_{r}
    &\le \frac{T}{q_r}
     \le T
     = \frac{n}{d+1}.
\end{split}
\end{equation}

Now consider a pool $P$ in round $r<d$. At most
$q_{r+1}\le T$ of its bins are selected by the lottery, while its group has
at most $\lceil n/(d+1)\rceil$ agents. Since $d\ge1$,
\begin{equation}
\label{eq:xos-offline-capacity}
    s_{r}+|Q_r\cap\{B_{P,1},\ldots,B_{P,n}\}|
    \le
   T+q_{r+1} \le 2T
    \le n.
\end{equation}

Give every agent of $P$ a distinct full bin not selected by $Q_r$.
No later pool is contained in such a bin. In the last round, give the agents
of each pool distinct full bins; there are enough because $s_{d}\le n$.

These allocations are pairwise disjoint. Pools in the same round are
disjoint, and any later pool lies in bins that were retained, rather than
allocated by this offline rule, in all earlier rounds. Each agent receives
value equal to her MMS. Any goods left over may be distributed arbitrarily
without decreasing values.
\end{proof}


Fix a deterministic online algorithm. Call an agent \emph{successful} if she
receives at least $\alpha$ times her MMS, and call a group successful if all
its agents are successful. 
Since an agent's value is the maximum contribution from a single bin, success
requires at least $\alpha M_r$ goods allocated from one bin in the agent's pool. Charge each successful agent at $P$ to the least-indexed bin
from which her allocation contains at least $\alpha M_r$ goods. Let $D_P$
be the set of charged bins; we call these bins {\em damaged}. The next lemma shows that each successful service of a pool $P$ will charge a lot of bins. 

\begin{lemma}
\label{lem:xos-pool-forcing}
If the group at a round-$r$ pool $P$ succeeds, then $|D_P|\ge\alpha\, s_{r}$.
Each bin in $D_P$ loses at least $\alpha M_r$ goods to this group, in addition to any goods allocated before the group arrived.
\end{lemma}

\begin{proof}[Proof of Lemma~\ref{lem:xos-pool-forcing}]
Let $k=\lceil\alpha M_r\rceil$. A successful agent has a bin contributing
at least $k$ goods to her allocation, so the prescribed charging rule is
well-defined. Since allocations are disjoint, a bin with $M_r$ goods receives
at most $M_r/k$ charges. Consequently,
\[
    s_{r}\le |D_P|\frac{M_r}{k}
    \le \frac{|D_P|}{\alpha}.
\]
Every bin in $D_P$ supplies at least $k\ge\alpha M_r$ goods to a charged
agent in the current group. These goods were available when allocated and
are therefore distinct from every good allocated before the group arrived.
\end{proof}

The next lemma shows that at least one of the damaged bins $D_P$ is retained for the next round with high probability.


\begin{lemma}
\label{lem:xos-continuation}
Let $r<d$, and consider any fixed round-$r$ pool $P$. Conditioned on the group at $P$ succeeding, we have
\[
    \Pr[D_P\cap Q_r=\varnothing]
    \le \exp\!\left(-\frac{\alpha K}{2(d+1)}\right).
\]
\end{lemma}

\begin{proof}[Proof of Lemma~\ref{lem:xos-continuation}]
Condition on a history as in the statement. The current pools, their bins,
the selected pool $P$, and its charged set $D_P$ are fixed. Conditional on the current bins, $Q_r$ is therefore still a
uniform $q_{r+1}$-element subset of the $nq_r$ bins.

On a history in which the group at $P$ succeeds,
Lemma~\ref{lem:xos-pool-forcing} gives $|D_P|\ge\alpha s_{r}$. Thus
\begin{align*}
    \Pr[D_P\cap Q_r=\varnothing\mid\text{history}]
    &=\frac{\binom{nq_r-|D_P|}{q_{r+1}}}
            {\binom{nq_r}{q_{r+1}}}\\
    &\le \left(1-\frac{|D_P|}{nq_r}\right)^{q_{r+1}}\\
    &\le \exp\!\left(-\frac{|D_P|q_{r+1}}{nq_r}\right)\\
    &\le \exp\!\left(-\frac{\alpha K}{2(d+1)}\right).
\end{align*}
If $nq_r-|D_P|<q_{r+1}$, the probability is zero and the bound is immediate;
otherwise the first inequality follows by multiplying the conditional miss
probabilities of successive draws without replacement. The last inequality
uses~\eqref{eq:xos-rounded-groups} and $q_{r+1}=Kq_r^2$. 
\end{proof}

We now consider the situation where a damaged round-$r$ bin becomes a pool in round $r+1$. The next lemma shows that when this pool is further partitioned randomly into $n$ bins, the items that it has lost to previous allocations also spread out near-uniformly across the partition with high probability. As a result, all of the bins in the pool are significantly depleted.


\begin{lemma}
\label{lem:xos-balance}
Let $P$ contain $nM$ goods, let $E\subseteq P$ be a fixed set of previously allocated
goods, and put $\rho=|E|/(nM)$. A uniform partition of $P$ into $n$ bins of
size $M$ satisfies, for every $\varepsilon>0$,
\[
    \Pr\!\left[\text{some bin }B\text{ has }
       |E\cap B|<(\rho-\varepsilon)M\right]
    \le \frac{n}{4\varepsilon^2M}.
\]
\end{lemma}

\begin{proof}[Proof of Lemma~\ref{lem:xos-balance}]
Fix one bin $B$ and let $X=|E\cap B|$. The bin is a uniform $M$-element
subset of $P$. Put $N=|P|=nM$ and $a=|E|$, so $\rho=a/N$.
View the bin as $M$ successive draws without replacement, and write
$X=\sum_{j=1}^M Z_j$, where $Z_j$ indicates that draw $j$ lies in $E$.
Each $Z_j$ has mean $\rho$ and variance at most $1/4$. If $N>1$, then for
distinct draws $j$ and $\ell$,
\[
    \mathbb E[Z_jZ_\ell]
       =\frac{a(a-1)}{N(N-1)}\le\left(\frac aN\right)^2.
\]
Their covariance is nonpositive, and hence
$\operatorname{Var}(X)\le M/4$. When $N=1$, $X$ is deterministic and the
same variance bound holds.

Chebyshev's inequality now gives
\[
    \Pr[X<(\rho-\varepsilon)M]
    \le \Pr[|X-\rho M|\ge\varepsilon M]
    \le \frac{\operatorname{Var}(X)}{\varepsilon^2M^2}
    \le \frac{1}{4\varepsilon^2M}.
\]
Taking a union bound over the $n$ bins proves the claim.
\end{proof}

Set $\varepsilon=\alpha/(2d)$. Motivated by the above lemma, for a round-$r$ pool $P$, with $E_P$ denoting the subset of $P$'s goods allocated by the end of round $r-1$, we say that the partition of $P$ into bins is {\em balanced} if in every bin at least a $(|E_P|/|P|-\varepsilon)$ fraction of the goods belong to $E_P$.  

Armed with the above lemmas, we will construct a sequence of pools, $P_0, P_1, \cdots$, that witness a failure event for the algorithm, namely an agent receiving less than $\alpha$ times her \MMS\ value. We start with $P_0=G$. In each round $r$, there are three possibilities. First, if the group at $P_r$ fails, that is, at least one agent in the group is unsuccessful, then we are already done. Second, if $D_{P_r}\cap Q_r=\emptyset$, then we terminate the witness sequence early without identifying a failure; we will bound the probability of such an event. Third, if $D_{P_r}\cap Q_r$ is non-empty, then we set $P_{r+1}$ to be the least-indexed bin in
$D_{P_r}\cap Q_r$ and continue to the next round.




The next lemma shows that if we succeed in constructing a witness sequence of length $d$ and every partition of the pools in the witness sequence is balanced, then the last witness pool in the sequence, $P_d$, is left with too few items to satisfy its last round agents. 


\begin{lemma}
\label{lem:xos-accumulation}
If the selected witness sequence reaches round $d$ and the bin partitions of
$P_0,\ldots,P_d$ are balanced, then every bin of $P_d$ has at most
$(\alpha/2)n^2$ goods available at the start of round $d$. Consequently,
no agent assigned to $P_d$ can receive $\alpha$ times her \MMS\ value.
\end{lemma}

\begin{proof}[Proof of Lemma~\ref{lem:xos-accumulation}]
Recall that $E_r$ records deletions in $P_r$ by the end of round $r-1$,
and $\rho_r=|E_r|/|P_r|$. We prove that, whenever the selected sequence has
reached round $r\ge1$ and all preceding selected refinements are balanced,
\begin{equation}
\label{eq:xos-path-induction}
    \rho_r\ge r\alpha-(r-1)\varepsilon.
\end{equation}
For $r=1$, the chosen $P_1$ is a charged bin of the successful round-$0$
group. At least $\alpha M_0$ of its $M_0$ goods were allocated to that
group, giving $\rho_1\ge\alpha$.

Suppose~\eqref{eq:xos-path-induction} holds for $r<d$. Balance implies that
every bin of $P_r$ contains at least $(\rho_r-\varepsilon)M_r$ goods of
$E_r$. The set $E_r$ is frozen at the previous round boundary; any additional
allocations before the group at $P_r$ arrives can only increase the number
of deleted goods. If the sequence continues, this group succeeds and
$P_{r+1}$ is one of its charged bins. Lemma~\ref{lem:xos-pool-forcing}
supplies at least $\alpha M_r$ newly allocated goods in this bin, all
disjoint from $E_r$. Since $|P_{r+1}|=M_r$, by the end of round $r$ its
deleted fraction is at least
\[
    \rho_{r+1}\ge\rho_r-\varepsilon+\alpha
       \ge(r+1)\alpha-r\varepsilon.
\]
This proves the induction.

At round $d$, a balanced refinement of $P_d$ therefore gives every bin a
deleted fraction at least
\[
    \rho_d-\varepsilon
       \ge d\alpha-d\varepsilon
       =d\alpha-\frac\alpha2
       \ge1-\frac\alpha2,
\]
since $d=\lceil1/\alpha\rceil$. Every last-round bin has size
$M_d=n^2$, so at most $(\alpha/2)n^2$ goods remain in each bin at the start
of this round. Subsequent allocations cannot increase this number. An agent
assigned to $P_d$ has MMS $n^2$ and values her allocation by its largest
intersection with a single bin. Even giving her all remaining goods in
$P_d$ yields value at most $(\alpha/2)n^2<\alpha n^2$.
\end{proof}

To complete our argument, we show that conditioned on the algorithm succeeding in rounds $1$ through $d-1$, with high probability, we can construct a complete witness sequence $P_0, \cdots, P_d$ such that the bin partitions of $P_0,\ldots,P_d$ are balanced.


\begin{theorem}
\label{thm:xos-distributional}
For every deterministic online allocation algorithm $A$,
\begin{equation*}
\label{eq:xos-pool-success}
\begin{split}
    & \Pr_{I\sim\mathcal D_{n,\alpha}}
       [A\text{ produces an }\alpha\text{-MMS allocation}]
    \le \eta_{n,\alpha},\\
    & \qquad \text{where} \qquad \eta_{n,\alpha}
    :=d\exp\!\left(-\frac{\alpha K}{2(d+1)}\right)
      +\frac{d^3}{\alpha^2n}.
\end{split}
\end{equation*}
In particular, if we take $\alpha=\frac{2}{\log\log n}$, we have $\eta_{n,\alpha}=o(1)$.
\end{theorem}

\begin{proof}[Proof of Theorem~\ref{thm:xos-distributional}]
Fix a deterministic online algorithm. Regard every random
choice in the instance as having been sampled before the first arrival.
Expose all partitions of a round to the analysis at the start of that round,
but expose its lottery only after the entire round has been allocated. Select
the next witness pool immediately after exposing that lottery, before
exposing any partition of the next round.

This scheme preserves the conditional laws needed by the lemmas. The
algorithm's allocations are functions only of the valuations it has actually
seen and its own earlier decisions. The witness-pool selection uses only
already exposed partitions, allocations, and lotteries. Consequently,
conditional on the auxiliary history before a lottery, that lottery is
still uniform over subsets of the current bins of the prescribed size.
Conditional on the auxiliary history before a new round, the selected pool
and its frozen deleted set are fixed, while its fresh partition is uniform.
Neither assertion conditions on eventual success or on an unexposed
partition.

If a group encountered by the selected sequence fails, simultaneous success
is already impossible. If each such group before round $d$ succeeds, each
encountered lottery retains a charged bin, and each selected refinement is
balanced, Lemma~\ref{lem:xos-accumulation} forces a last-round agent to fail.
That pool has a nonempty group by Lemma~\ref{lem:xos-instances}. Therefore,
simultaneous success is contained in the union of the events that an
encountered lottery misses the charged set of a successful selected group or
an encountered refinement is not balanced.

There are at most $d$ events of each kind. For each lottery, on every history
where that step is reached with a successful group,
Lemma~\ref{lem:xos-continuation} bounds its conditional miss probability by
$\exp(-\alpha K/(2(d+1)))$. For a selected round-$r$ refinement,
Lemma~\ref{lem:xos-balance} bounds its conditional failure probability by
$n/(4\varepsilon^2M_r)$. These are bounds on the steps of the selected
sequence, not on all pools in the round. Averaging each conditional bound and
taking a union bound yields
\begin{align*}
    \Pr[A\text{ produces an }\alpha\text{-MMS allocation}]
    &\le d\exp\!\left(-\frac{\alpha K}{2(d+1)}\right)
         +\sum_{r=1}^d\frac{n}{4\varepsilon^2M_r}\\
    &\le d\exp\!\left(-\frac{\alpha K}{2(d+1)}\right)
         +\frac{dn}{4\varepsilon^2n^2}\\
    &=d\exp\!\left(-\frac{\alpha K}{2(d+1)}\right)
         +\frac{d^3}{\alpha^2n}.
\end{align*}
The last two inequalities hold because $M_r\ge n^2$ and $\varepsilon=\alpha/(2d)$. Finally, since $\alpha = \frac{2}{\log\log n}$ and $K\ge 2\alpha^2\log n$, we know $d\exp\!\left(-\frac{\alpha K}{2(d+1)}\right)\le d\exp\!\left(-\frac{\alpha^3 \log{n}}{d+1}\right)=O(\log\log n)\exp\!\left(-\Omega(\frac{\log n}{(\log\log n)^4})\right)=o(1)$ and $\frac{d^3}{\alpha^2n}=O(\frac{(\log\log n)^5}{n})=o(1)$. So the RHS is $o(1)$.
\end{proof}



The Yao reduction now proves Theorem~\ref{thm:xos-lower-bound}.

\section{Failure of Negative Association under Iterative SwR}
\label{app:negassoc}

We show the failure of negative association (NA) under different sampling without replacement processes. To keep the exposition conceptually cleaner and intuitive in this section, we strip down our algorithms to the minimal processes of sampling without replacement with fixed probabilities under fixed and adaptive bucketing. First, in \ref{sec:adapprob} we show that adaptive allocation probabilities can break pairwise negative dependence. In subsection \ref{app:negassoc-additive-fixed}, we show that even when all probabilities are fixed, a sampling-without-replacement process can fail to be negatively associated after two rounds given fixed buckets. Finally, in subsection \ref{app:subsec1}, we show that under adaptive bucketing, a sampling-without-replacement process can fail to be negatively associated after just one round.

Recall that $\{0,1\}$-random variables $X_1, \cdots, X_n$ are said to be  negatively associated (NA) if for all subsets $A,B\subset [n]$ with $A\cap B=\emptyset$ and for all component-wise non-decreasing functions $f:\{0,1\}^A\rightarrow \mathbb{R}$ and $g:\{0,1\}^B\rightarrow \mathbb{R}$, we have $\operatorname{Cov}(f(X_i, i\in A), g(X_j, j\in B))\le 0$.

\subsection{SwR with adaptive probabilities fails pairwise negative dependence}
\label{sec:adapprob}

We first show that adaptive sampling probabilities can violate pairwise negative correlation even within a single round. Sample item $1$ with probability $p\in(0,1)$, and then sample item $2$ with probability $p$ if item $1$ was sampled and $0$ otherwise. The survival indicators $X_1,X_2$ satisfy $\mathbb E[X_1]=1-p$, $\mathbb E[X_2]=1-p^2$, and $\mathbb E[X_1X_2]=1-p$, giving $\operatorname{Cov}(X_1,X_2)=p^2(1-p)>0$.

\subsection{SwR with fixed probability and fixed buckets fails NA}
\label{app:negassoc-additive-fixed}

Consider two agents and a simple sampling without replacement
process in which each agent partitions all items into two fixed
buckets in advance. Once an agent arrives, we sample from the
remaining items in each bucket $B$ using the same fixed-probability
sampling rule $\swr(B,1/4)$.

Let letters $a$ to $f$ denote six available items, and let
\[
  U=\{b,c\},\qquad V=\{d,e,f\},\qquad
  G=\{a\}\cup U\cup V.
\]
Let the two buckets created by the first agent be $M_{1,1}=U$
and $M_{1,2}=\{a\}\cup V$, and those created by the second agent
be $M_{2,1}=U\cup V$ and $M_{2,2}=\{a\}$.\footnote{
These buckets can correspond to MMS-defining bundles under
additive valuations. For example, let the first agent's singleton
values be $v_1(\{x\})=1/2$ for $x\in U$ and
$v_1(\{x\})=1/4$ for $x\in\{a\}\cup V$, and let the second
agent's singleton values be $v_2(\{a\})=1$ and
$v_2(\{x\})=1/5$ for $x\in U\cup V$.
For each agent the MMS is exactly $1$.}

Let $G^{(r)}$ denote the set of remaining items after the arrival
of the $r$th agent, where $G^{(0)}=G$. In round $r$, the buckets
are $M_{r,j}\cap G^{(r-1)}$, and sampling from them is independent
conditioned on the history. Recall that for a given bucket $B$, $SwR(B,1/4)$ samples $k$ from $\{\lceil \frac{|B|}{4} \rceil, \lfloor \frac{|B|}{4} \rfloor\}$ such that its expectation is exactly $\frac{|B|}{4}$, and then samples a uniform random subset of size $k$ from $B$. An empty bucket contributes no
sampled items.
In round $1$, exactly one item is sampled from $\{a\}\cup V$.
From $U$, no item is sampled with probability $1/2$, and one
uniformly chosen item is sampled with probability $1/2$.
Let us define
\[
\mathcal{E}=\mathbf{1}\{a\in G^{(1)}\},
  \qquad
  \mathcal{F}=\mathbf{1}\{a\in G^{(2)}\},
  \qquad
  \mathcal{H}=\mathbf{1}\{\{b,c\}\cap G^{(2)}\neq\emptyset\}.
  \]
Event $\mathcal F$ indicates that item $a$ remains, while $\mathcal{H}$ indicates that at least one of $b,c$ remains, after round 2. Note that $\mathcal{F}$ and $\mathcal{H}$ are coordinate-wise increasing functions of disjoint sets of survival indicators. We will show that the covariance between $\mathcal{F}$ and $\mathcal{H}$ is strictly positive, violating NA.

The probability of event $\mathcal{E}$, that item $a$ remains after round 1, is $3/4$. Conditioned on $\mathcal{E}$, that is, item $a$ reaches round $2$, it survives after round 2 
with probability $3/4$; on $\mathcal{E}^c$, it has already
been removed. Consequently,
\[
  \mathbb{E}[\mathcal{F}\mid\mathcal{E}]=\frac34,
  \qquad
  \mathbb{E}[\mathcal{F}\mid\mathcal{E}^c]=0.
\]

Let $Q=|\{b,c\}\cap G^{(1)}|$. Its distribution is independent
of $\mathcal{E}$, and
\[
  \Pr(Q=1)=\frac12,\qquad \Pr(Q=2)=\frac12.
\]

Conditioned on $\mathcal{E}$, two of the three items of $V$ survive round $1$. Bucket $B=M_{21}\cap G^{(1)}$, therefore, has size three or four, and hence at most one item is sampled from it.
If $Q=2$, at least one of $b,c$ survives surely; if $Q=1$,
the sole remaining distinguished item survives with
probability $3/4$. Hence
\[
  \mathbb{E}[\mathcal{H}\mid\mathcal{E}]
  =\frac12\cdot\frac34+\frac12
  =\frac78.
\]

While, conditioned on $\mathcal{E}^c$, $V\subset B$ since all of its items survive round $1$. Now, if $Q=1$, the sole remaining item of $\{b,c\}$ survives round $2$ with probability $3/4$. If $Q=2$, then $B$ contains five items, and hence $SwR(B,1/4)$ samples two items from it with probability $1/4$, and conditional on that sample size,
the sampled pair is $\{b,c\}$ with probability $1/\binom52$.
Thus both $b,c$ are removed with probability $1/40$, giving
\[
  \mathbb{E}[\mathcal{H}\mid\mathcal{E}^c]
  =\frac12\cdot\frac34+\frac12\cdot\frac{39}{40}
  =\frac{69}{80}.
\]

Conditional on $\mathcal{E}$ or $\mathcal{E}^c$, the variables
$\mathcal{F}$ and $\mathcal{H}$ are independent: on $\mathcal{E}$,
the survival of $a$ depends only on the independent sampling
from its singleton bucket, and on $\mathcal{E}^c$,
$\mathcal{F}=0$. Therefore,
\[
  \mathbb{E}[\mathcal{F}]=\frac9{16},
  \qquad
  \mathbb{E}[\mathcal{H}]
  =\frac34\cdot\frac78+\frac14\cdot\frac{69}{80}
  =\frac{279}{320}.
\]
Also,
\[
  \mathbb{E}[\mathcal{F}\mathcal{H}]
  =\frac34\cdot\frac34\cdot\frac78
  =\frac{63}{128}.
\]
Hence,
\[
  \operatorname{Cov}(\mathcal{F},\mathcal{H})
  =\frac{63}{128}-\frac9{16}\cdot\frac{279}{320}
  =\frac9{5120}>0.
\]
This shows failure of negative association after two rounds of the simplest SwR process, namely with sampling probabilities and fixed buckets.

\subsection{SwR with fixed probability and adaptive buckets fails NA}
\label{app:subsec1}
Consider six items and a simple sampling without replacement process in which an agent iteratively creates a bucket containing the two items with the highest marginal values relative to its current set of sampled items, starting from the empty set. Every bucket $B$ uses the same fixed-probability sampling rule $\swr(B,1/2)$. Since each bucket contains two items, exactly one item is sampled uniformly from each bucket and added to the agent's set. The agent then creates the next bucket from the remaining unconsidered items. Under a submodular valuation, the marginal values of these items, and hence the items of the next bucket, may depend on the item sampled from the first bucket. Thus, the next bucket is constructed adaptively based on the sampled items.


Let $G^0=\{1,\cdots,6\}$ denote the set of initial items, and let the initial bucket that agent created be
$B_0=\{1,2\}$. Define
\[
  \mathcal{E}=\{\text{item }1\text{ is sampled from }B_0\}.
\]
Both $\mathcal{E}$ and its complement $\mathcal{E}^c$ have
probability $1/2$. We will later provide a submodular valuation function such that based on the event $\mathcal{E}$, the subsequent buckets are as follows:
\[
\begin{array}{c|cc}
 & B_1 & B_2 \\ \hline
\mathcal{E}   & \{3,5\} & \{4,6\} \\
\mathcal{E}^c & \{3,4\} & \{5,6\}
\end{array}.
\]
Let $G^1$
denote the set of remaining items after all
three buckets have been processed and the sampled items removed.

The key observation is that item $2$ remains precisely on
$\mathcal{E}$. On this event, items $3$ and $4$ occupy different
buckets and can both remain. On $\mathcal{E}^c$, they occupy
the same bucket and cannot both remain. To formalize this, define
\[
  \mathcal{F}=\mathbf{1}\{2\in G^1\},
  \qquad
  \mathcal{H}=\mathbf{1}\{3\in G^1\}\mathbf{1}\{4\in G^1\}.
\]
Conditional on $\mathcal{E}$,
we have $\mathcal{F}=1$, and items $3$ and $4$ remain independently
with probability $1/2$ each. Hence
\[
  \mathbb{E}[\mathcal{F}\mid\mathcal{E}]=1,
  \qquad
  \mathbb{E}[\mathcal{H}\mid\mathcal{E}]
  =\mathbb{E}[\mathcal{F}\mathcal{H}\mid\mathcal{E}]
  =\frac14.
\]
Conditional on $\mathcal{E}^c$, we have
$\mathcal{F}=\mathcal{H}=0$. Therefore
\[
  \mathbb{E}[\mathcal{F}]=\frac12,
  \qquad
  \mathbb{E}[\mathcal{H}]=\frac18,
  \qquad
  \mathbb{E}[\mathcal{F}\mathcal{H}]=\frac18,
\]
and consequently
\[
  \operatorname{Cov}(\mathcal{F},\mathcal{H})
  =\mathbb{E}[\mathcal{F}\mathcal{H}]
   -\mathbb{E}[\mathcal{F}]\mathbb{E}[\mathcal{H}]
  =\frac18-\frac12\cdot\frac18
  =\frac1{16}>0.
\]
Thus the survival indicators
$\bigl(\mathbf{1}\{j\in G^1\}\bigr)_{j\in G}$
are not negatively associated.

\paragraph{A submodular valuation realizing these buckets.}
Let $R=\{3,4,5,6\}$ and fix
$0<\varepsilon<1/10$. Define
\[
  w_1=\frac12+2\varepsilon,
  \qquad
  w_2=\frac12+3\varepsilon,
  \qquad
  w_k=\frac12 \quad(k\in R),
\]
and
\[
\begin{array}{c|rrrr}
  k      &3&4&5&6\\ \hline
  c_{1k} &0&2&1&3\\
  c_{2k} &0&1&3&2
\end{array}.
\]
For every $S\subseteq G$, set
\begin{equation}
\label{eq:adaptive-valuation}
  v_1(S)
  =\sum_{j\in S}w_j
   -\varepsilon
    \sum_{r\in S\cap\{1,2\}}
    \sum_{k\in S\cap R}c_{rk}.
\end{equation}

This valuation produces precisely the buckets above as initially, items $1,2$ have the largest singleton values. After sampling $1$, the remaining marginal order is $3,5,4,6$, and after sampling $2$, the remaining marginal order is $3,4,6,5$.

\section{Appendix for XOS Upper Bound Analysis}\label{app:xos-upper-bound}
We show the complete proof of Theorem \ref{thm:xos-alg} here.
\begin{proof}[Proof of Theorem \ref{thm:xos-alg}]
    As showed in the proof previously, we want to bound the probability of complement of the event $\mathcal{E}$,
\[
    \mathcal{E}:=
    \{Z\le Q/{50}\}
    \cap
    \bigcap_{i=1}^n\{X_i\le {6q_i/25}\}.
\]
For the first term, we use the bound for $\Pr[Z > Q/50] \leq 101p$ from Lemma \ref{lem:xos-virtual-coupling}. For the second term of complement of $\mathcal{E}$, we bound the failure probability differently for agents with different $q_i$. Recall that $x = 1/(4\tau) \geq \max\{1500, \log n\}$. We split the analysis into three cases.
\begin{itemize}[leftmargin=*]
    \item \textbf{Case 1: $q_i \leq k/16$.} In Lemma \ref{lem:xos-weighted-loss},  we choose $\lambda = \log 4$ and apply Markov's inequality to the random variable $4^{X_i} - 1$. Note that, $\theta(\lambda) = pW(e^\lambda - 1) = 3pW \leq 3q_i/k$. We therefore get,
    \[
    \Pr[X_i > 6q_i/25] \leq \frac{\exp(2(e^{3q_i/k} - 1)) - 1}{4^{6q_i/25} - 1}.
    \]
    For the numerator, define $f(u) = \exp(2(e^u - 1)) - 1$. Since $f$ is convex and the end points in this case are $f(0) = 0$, $f(3(k/16)/k) \leq f(3/16)$, we can approximate the range with the intercept and obtain,
    \[
    \Pr[X_i > 6q_i/25] \leq \frac{\frac{f(3/16)}{k/16} q_i}{4^{6q_i/25} - 1}.
    \]
    We can therefore bound the failure probability in this case as,
    \[
    \sum_{i: q_i \leq k/16} \Pr[X_i > 6q_i/25] \leq \frac{16f(3/16)}{k} \sum_{i:q_i \leq k/16} \frac{q_i}{4^{6q_i/25} - 1}.
    \]
    Now, $\frac{q_i}{4^{6q_i/25} - 1}$ decreases as a function of $q_i$. Additionally, we know that $q_i \geq n - i + 1$. Therefore,
    \[
    \sum_{i: q_i \leq k/16} \frac{q_i}{4^{6q_i/25} - 1} \leq \sum_{r \geq 1} \frac{r}{4^{6r/25} - 1} < 13.4.
    \]
    Therefore, together, we get in this case,
    \[
    \sum_{i: q_i \leq k/16}\Pr[X_i > 6q_i/25] < \frac{110}{k}.
    \]
    \item \textbf{Case 2: $k/16 < q_i < 1000 x$} This time we use $\lambda = k/16q_i$ in Lemma \ref{lem:xos-weighted-loss} and apply Markov's inequality to the random variable $e^{\lambda X_i}$. We therefore get,
    \[
    \Pr[X_i \geq 6q_i/25] \leq \frac{\exp (2(e^{pq_i(e^\lambda - 1)} - 1)) }{\exp(6\lambda q_i/25)}
    \]
    \[
    \Pr[X_i > 6q_i/25] \leq \exp\left(2(e^{(e-1)/16} - 1) - \frac{3k}{200}\right) < e^{0.227-3k/200}.
    \]
    Again, since $n-i+1 \leq q_i < 1000x$, any agent $i$ in this case must arrive at position $i > n - 1000 x + 1$. Therefore, there are at most $1000 x$ many agents of this type. Therefore,
    \[
    \sum_{i: k/16 < q_i < 1000 x}\Pr[X_i > 6q_i/25] < 1000 x e^{0.227 - 3k/200}.
    \]
    Now, consider the RHS of the above term. Recall that $k= \lfloor x \rfloor$. We can verify using simple algebra that
    \[
    1000 xe^{0.227 - 3k/200} \leq \frac{1}{x}
    \]
    for all $x \geq 1500$, and therefore,
    \[
    \sum_{i: k/16 < q_i < 1000 x}\Pr[X_i > 6q_i/25] \leq \frac{1}{x}.
    \]
    \item \textbf{Case 3: $q_i \geq 1000 x.$} In this case, we choose $\lambda = 5.2 x/q_i \leq 0.0052$. Again, apply Markov on the variable $e^{\lambda X}$. Using Lemma \ref{lem:xos-weighted-loss},
    \[
    \Pr[X_i > 6q_i/25] \leq \frac{\exp(2(e^{pq_i(e^\lambda - 1)} - 1))}{\exp(6\lambda q_i/25)}.
    \]
    Now, $\lambda$ decreases as a function of $q_i$ and at $q_i = 1000x, \lambda = 0.0052$. Therefore, using convexity of $e^\lambda - 1$,
    \[
    pq_i(e^\lambda - 1) \leq pq_i(e^{0.0052} - 1)\frac{\lambda}{0.0052} = 1000x p(e^{0.0052} - 1).
    \]
    Substituting in the previous inequality, we get
    \[
    \Pr[X_i > 6q_i/25] \leq e^{367.41 - 6\lambda q_i/25} = e^{367.41 - 156x/125}.
    \]
    The total number of agents in this case is bounded by $e^x$. Therefore, the total probability of failure for all agents that are in this case is at most $e^{367.41 - 31x/125}$. Again, with elementary algebra we can verify that,
    \[
    e^{367.41 - 31x/125} \leq \frac{16}{x}.
    \]
\end{itemize}
Therefore,
\begin{align*}
\sum_{i \in [n]} \Pr[X_i > 6q_i/25] &\leq \frac{110}{\lfloor x \rfloor} + \frac{17}{x} \leq \frac{128}{x}.
\end{align*}

From Lemma \ref{lem:xos-virtual-coupling}, we know that $\Pr[Z > Q/50] \leq 101p$. Combining, we get that
\begin{align*}
    \Pr[\text{Algorithm \ref{alg:XOS} reports failure}] &\leq 101p+ \frac{128}{x} \\
    &\leq \frac{230}{x} \leq 1000 \tau.
\end{align*}
Together with Lemma~\ref{lem:xos-allocation-value}, this proves
the theorem.

\end{proof}

\end{document}